\documentclass[letterpaper, 10 pt, conference]{ieeeconf} 
\IEEEoverridecommandlockouts  
\usepackage{amsmath}
\usepackage{amssymb}
\usepackage{amsthm}
\usepackage{mathtools}
\usepackage{derivative}
\NewDerivative{\tpdv}{\partial}[style-frac=\tfrac]
\usepackage{xcolor}
\usepackage{bm}

\usepackage[noadjust]{cite}
\usepackage{url}

\usepackage{graphicx}
\usepackage[export]{adjustbox} % For aligning figures
\graphicspath{{figures/}}

\newtheorem{theorem}{Theorem}
\newtheorem{lemma}{Lemma}

\theoremstyle{definition}
\newtheorem{definition}{Definition}

\newtheorem{remark}{Remark}

\newcommand{\R}{\mathbb{R}}

\newcommand{\T}{^\top}

\newcommand{\ba}{\mathbf{a}}
\newcommand{\bb}{\mathbf{b}}

\newcommand{\be}{\mathbf{e}}
\renewcommand{\bf}{\mathbf{f}} % NOTE: use \textbf if you really want to write bold non-math text.
\newcommand{\bg}{\mathbf{g}}

\newcommand{\bk}{\mathbf{k}}

\newcommand{\bp}{\mathbf{p}}

\newcommand{\bu}{\mathbf{u}}
\newcommand{\bv}{\mathbf{v}}

\newcommand{\bx}{\mathbf{x}}

\newcommand{\bA}{\mathbf{A}}
\newcommand{\bB}{\mathbf{B}}

\newcommand{\bG}{\mathbf{G}}

\newcommand{\bI}{\mathbf{I}}
\newcommand{\bJ}{\mathbf{J}}

\newcommand{\bM}{\mathbf{M}}

\newcommand{\bR}{\mathbf{R}}

\newcommand{\bxi}{\bm{\xi}}

\newcommand{\bomega}{\bm{\omega}}

\newcommand{\bnu}{\bm{\nu}}

\usepackage{xcolor}
\definecolor{myblue}{RGB}{49, 114, 174}
\definecolor{myred}{rgb}{0.796, 0.235, 0.2}
\definecolor{mygreen}{rgb}{0.22, 0.596, 0.149}
\definecolor{mypurple}{rgb}{0.584,0.345,0.698}

\usepackage{blindtext}
\title{\textbf{Vector Field Following for Quadrotors via Reduced Attitude Representations with Applications to Safety-Critical Control}}

\author{Mohammad Mirtaba and Max H. Cohen%
\thanks{Department of Electrical and Computer Engineering, North Carolina State University, Raleigh, NC, \texttt{\{mmirtab,mhcohen2\}@ncsu.edu}.}
}

\begin{document}
\maketitle
\begin{abstract}
    Vector field tracking for quadrotors provides a systematic framework for specifying complex behaviors at the translational level and transferring them to the full vehicle dynamics. Objectives such as safety, for example, can be encoded into a low-dimensional vector field, 
    which is then embedded into the full dynamics through appropriate tracking conditions. 
    The development of such tracking controllers for quadrotors can be facilitated by the fact that their translational dynamics depend only on the thrust direction and are independent of the vehicle's heading. Thus, reduced-attitude representations, leveraging the geometry of the unit sphere, provide a natural framework for vector field tracking on quadrotors. In this work, we propose two vector field tracking controllers for quadrotors using reduced-attitude representations based on dynamic extension and geometric control, respectively. Next, we illustrate how safety can be encoded into the desired vector field using control barrier functions and enforced for the full system. The effectiveness of the proposed approaches is demonstrated via simulations and experimental flight tests.
\end{abstract}

\section{Introduction}
% The use of quadrotors has expanded across a wide range of applications and has reached the commercial market making them readily available for consumer use in areas such as package delivery, inspection, and surveillance. 
Although commercially available quadrotors and flight controllers provide satisfactory flight stability and maneuverability, ensuring safety during operation remains a significant challenge. Quadrotors are increasingly deployed in urban environments, often in cooperation with other robotic systems. Thus, imposing safety constraints is essential to prevent hazardous flight operations that could endanger nearby individuals, other robots, or surrounding infrastructure.

% Quadrotor control techniques initially relied on Euler angle representations and independent proportional-integral-derivative (PID) controllers for each control loop in a cascaded attitude-position architecture \cite{bouabdallah2007full}. These methods later leveraged the geometry of $\mathrm{SO}(3)$, the space of rotation matrices, to avoid singularities associated with Euler-angle representations, and to provide almost global stability \cite{lee2010geometric}.

Controllers leveraging attitude dynamics on $\mathrm{SO}(3)$~\cite{lee2010geometric} require regulation of the entire orientation, including the heading direction (the rotation about the thrust vector), even though the heading does not directly influence the translational dynamics. Here, heading control authority is generated through differential rotor reaction torques, which are weaker than the moments used for tilting maneuvers~\cite{gamagedara2019geometric}.

Aggressive heading tracking may impose additional demands on actuators that could lead to saturation and degradation of translational stability.
% particularly during aggressive maneuvers. 
Since the thrust direction is the primary variable governing the translational dynamics, reduced-attitude approaches, leveraging the geometry of the unit sphere, have been developed to stabilize the thrust direction while leaving the heading 
% degree of freedom 
unconstrained~\cite{leomanni2024almost}.

Although recent reduced-attitude control approaches decouple translational dynamics from heading control, most existing designs rely on time-scale separation arguments, requiring significantly faster attitude tracking. This limitation, commonly encountered in cascaded controllers, has also been observed in full attitude control formulations on $\mathrm{SO}(3)$,
\begin{figure}[t]
    \centering
    \vspace{2.5mm}
    \includegraphics[width=1\columnwidth]{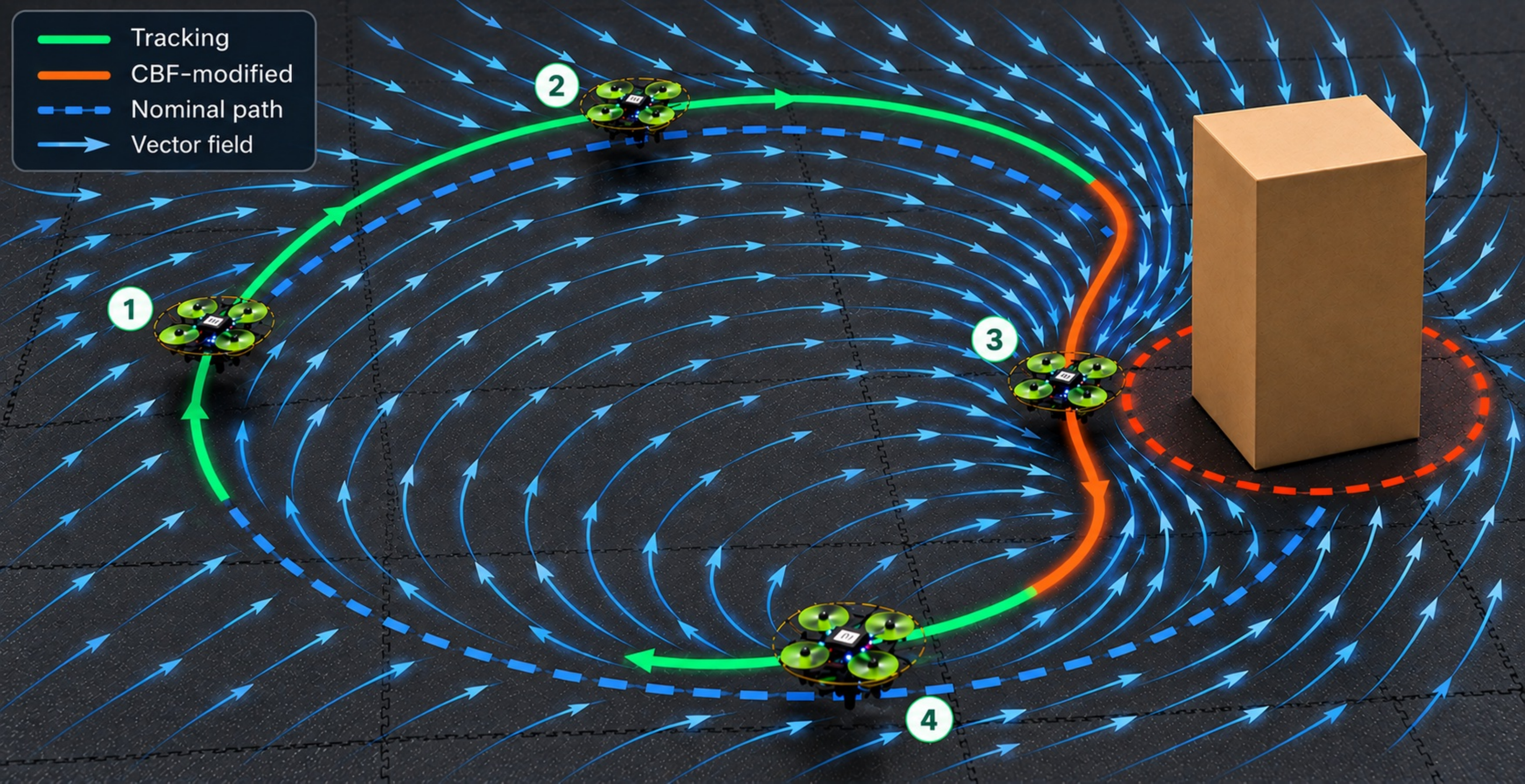}
    \vspace{-6mm}
    \caption{Safety-critical vector field following for quadrotors via reduced attitude representations. Video available at https://youtu.be/Fn9biWWtrd8.}
    \label{fig:exp_video}
    \vspace{-6mm}
\end{figure}
which has motivated recent research that explicitly considers the coupling between translational and rotational dynamics. Here, stability of the resulting cascade is often analyzed by treating the attitude tracking error as an input to the translational subsystem, and requiring the translational dynamics to satisfy input-to-state stability (ISS) \cite{roza2014class} or integral ISS (iISS) \cite{invernizzi2019integral} properties. Although these properties are powerful for establishing stability of interconnected systems, they mainly offer qualitative guarantees
in the sense that
% . As a result, 
obtaining quantitative performance measures, such as convergence rates, can be challenging \cite{sontag2008input}. Moreover, enforcing such conditions may limit the selection of more aggressive position control laws and potentially affect transient performance.

The literature on quadrotor control is now quite mature, but much research has focused on stabilization tasks. As noted earlier, however, recent advancements necessitate the consideration of safety specifications. 
% such as the strict satisfaction of state constraints.
Safety-critical control designs for quadrotors are facilitated by techniques such as reachability analysis \cite{reach_safe} and model predictive control \cite{mpc_safe}, which typically incur higher computational costs, and control barrier functions (CBFs) \cite{AmesTAC17}.
% which can inherit many of their underlying limitations. 
Examples of these aforementioned approaches include trajectory optimization enforcement methods \cite{landry2016aggressive} that require high computational cost and neglect low level dynamics, model predictive path integral based method \cite{jin2023safety}, CBFs designed for planar quadrotors \cite{planerquad}, leveraging cascaded architectures for safety \cite{cascadequad}, and exploiting the equivalence between integrator chains and quadrotor dynamics via differential flatness \cite{doeser2020invariant}.

% \MC{Think we need to add another paragraph here on vector field tracking and its connections to CBFs. Things to discuss here would be \cite{SchwagerICRA14,BahatiCDC25}, among others.}

Vector field tracking in robotic systems provides a means of specifying desired behaviors at a higher level and transferring these behaviors to systems with complex dynamics \cite{SchwagerICRA14}. This perspective is particularly relevant to safety-critical control, where safety-preserving behaviors can be synthesized on simplified models and subsequently realized by the full-order system through rigorous tracking guarantees. For example, \cite{TamasRAL22} synthesize CBFs on reduced-order models and then lift the resulting safety guarantees to the full nonlinear system by certifying convergence of the nonlinear dynamics to the reduced-order model (i.e., the vector field).

In this paper, we leverage reduced-attitude representations to design quadrotor controllers that enforce exponential convergence of the translational dynamics to a desired \emph{vector field} that carries safety properties. In particular, we 1) present two designs based on dynamic extension (DE) and geometric control, respectively; 2) explicitly address translational--rotational coupling terms that arise in existing formulations, which present roadblocks to rigorous stability proofs; 3) exploit the resulting exponential tracking guarantees to incorporate safety constraints into the quadrotor control architecture. In addition to theoretical results, we deploy the proposed controllers on a Crazyflie in experimental flight tests.

\section{Preliminaries and Problem Formulation}
In this section, we review quadrotor dynamics, their reduced-attitude representations, and CBFs.
% \subsection{Notation}
Herein, $\mathrm{SO}(3)$ denotes the special orthogonal group $
\mathrm{SO}(3) \coloneqq \{ \bR \in \mathbb{R}^{3\times 3} \mid \bR^\top \bR = \bI,\ \det(\bR)=1 \}
$ and $S^2\coloneqq \{\bx\in\R^3\mid\|\bx\|=1\}$ is the unit sphere. The Lie algebra of $\mathrm{SO}(3)$ is
$
\mathfrak{so}(3) \coloneqq \{ \bA \in \mathbb{R}^{3\times 3} \mid \bA^\top = -\bA \}
$. The operator $[\cdot]_\times:\mathbb{R}^3\rightarrow\mathfrak{so}(3)$ maps a vector to a skew-symmetric matrix, satisfying $
[\ba]_\times\bb=\ba\times\bb, \forall\,\ba,\bb\in\mathbb{R}^3.$ We use $\|\cdot\|$ for the Euclidean norm.

\subsection{Full Quadrotor Model}

The states of a quadrotor are typically described by $\bx=(\bp,\bv,\bR,\bomega)$, where $\bp,\bv\in\mathbb{R}^3$, $\bR\in \mathrm{SO}(3)$, and $\bm\omega \in \mathfrak{so}(3)$ denote the position and velocity of the center of mass, the attitude, and angular velocity of the quadrotor in the body-fixed frame, which lies in the Lie algebra of $\mathrm{SO}(3)$. The attitude is represented by a rotation matrix $\bR\in\mathrm{SO}(3)$. A $3$D vector can be associated with each element of the Lie algebra hence, we represent the body-fixed frame angular velocity as $\bm \omega \in \mathbb{R}^3$, and the inertial frame angular velocity as $\bomega^\mathrm{i} \coloneqq \bR \bomega\in \mathbb{R}^3$. The control input $\bu=(T,\bM)$ consists of the total thrust $T\in\mathbb{R}$, and the moments $\bM\in\mathbb{R}^3$ in the body-fixed frame. The dynamics of the quadrotor are then:
\begin{equation}
\label{eq:full_model}
 \dot{\bx}=
 \begin{bmatrix}
     \dot{\bp} \\ \dot{\bv} \\ \dot{\bR} \\ \dot{\bomega}
 \end{bmatrix}
 =
 \left[\begin{array}{c}
\bv \\
\frac{1}{m}\bR\bb_3T + \bg \\
\bR[\bm \omega]_{\times} \\
\bJ^{-1} (\bM - \bm\omega \times \bJ\bm \omega )
\end{array}\right].
\end{equation}
% where, for brevity of notation, we will sometimes use:
% \begin{align}
%     f &\coloneqq \frac{T}{m},\quad
%     \bm \tau \coloneqq \bJ^{-1} (\bM - \bm\omega \times \bJ\bm \omega ).
% \end{align}
where $\bb_3=(0,0,1)\in\mathbb{R}^3$ and $\bg = -g\bb_3\in\mathbb{R}^3$ is the gravitational acceleration vector, where $g>0$ denotes the gravitational acceleration magnitude. Furthermore, $m\in\mathbb{R}_{>0}$ is the vehicle mass, and $\bJ\in\mathbb{R}^{3\times3}$ is the positive-definite inertia matrix expressed in the body-fixed frame.

\subsection{Reduced Attitude Representation}
The translational dynamics $(\dot{\bp}, \dot{\bv})$ depend on the attitude only through the quantity $\bR\bb_3$, which corresponds to the body-fixed thrust axis expressed in the inertial frame. Consequently, the translational motion is independent of rotations about the thrust axis, motivating a reduced-attitude representation in which the orientation is described solely by the thrust direction, potentially simplifying control design. The thrust directions belongs to the unit sphere $S^2$, and represented as a unit vector in $\bR^3$. Accordingly, we define:
\begin{align}
 \bR_3 &\coloneqq \bR\bb_3,\\
 \dot \bR_3 &\coloneqq \bR[\bm\omega]_\times \bb_3=\bR_1 \omega_2 - \bR_2 \omega_1,
\end{align}
where $\bR_i$ is the $i$-th column of the rotation matrix and $\omega_i$ is the $i$-th element of the angular velocity vector. The dynamics expressed in terms of the reduced attitude with the  state $\bx =(\bp, \bv, \bR_3, \bm \omega)\in\R^3\times\R^3\times S^2\times \R^3$ are:
\begin{equation}
    \label{eq:reduced_model}
     \dot{\bx}=
     \begin{bmatrix}
     \dot{\bp} \\ \dot{\bv} \\ \dot{\bR} \\ \dot{\bomega}
     \end{bmatrix}=
     \left[\begin{array}{c}
    \bv \\
     \frac{1}{m}\bR_3T + \bg \\
    -\bR_2(t) \omega_{1}+ \bR_1(t) \omega_{2} \\
    \bJ^{-1} (\bM - \bm\omega \times \bJ\bm \omega )
    \end{array}\right].
\end{equation}
The vectors $\bR_1$ and $\bR_2$ satisfy the rotational kinematics of the full quadrotor model. Since the heading dynamics are decoupled from the reduced-order system and are assumed to be generated by an independent heading controller, $\bR_1$ and $\bR_2$ are treated as known time-varying variables.

In this paper, we treat the angular velocity $\bomega$, rather than the moment $\bM$, as a control input, yielding a system with state $\bx=(\bp,\bv,\bR_3)\in\R^3\times\R^3\times S^2$ and dynamics:
\begin{equation}
\label{eq:reduced_model_omega}
 \dot{\bx}=
 \begin{bmatrix}
     \dot{\bp} \\ \dot{\bv} \\ \dot{\bR}_3
 \end{bmatrix}=
 \left[\begin{array}{c}
\bv \\
 \frac{1}{m}\bR_3T + \bg \\
-\bR_2(t) \omega_{1}+ \bR_1(t) \omega_{2}
\end{array}\right],
\end{equation}
with control input $\bu=(T,\omega_{1},\omega_{2})$.

% \MC{I think we should have two more subsections here. One providing some preliminaries on CBFs and safety-critical control. One that discusses the vector field tracking problem, explains why we we are considering this vector field paradigm rather than traditional time-based references, and a formal problem statement, something along the lines of "given a smooth desired vector field, design a controller such that $v$ approaches this desired vector field.}

\subsection{Control Barrier Functions}
Consider a nonlinear time-varying control affine system:
\begin{align}
\dot{\bx} = \bf(\bx ,t) + \bG(\bx, t)\bu,
\label{eq:general_knowen_nonlinear_system}
\end{align}
with state $\bx \in \mathbb{R}^n$, control input $\bu \in \mathbb{R}^m$, drift dynamics $\bf:\mathbb{R}^n\times\R \rightarrow \mathbb{R}^n$, and control directions $\bG:\mathbb{R}^n \times\R\rightarrow \mathbb{R}^{n\times m}$.
If these dynamics are locally Lipschitz in $\bx$ uniformly in $t$, then any control $\bk:\R^n\times\R\rightarrow\R^m$, locally Lipschitz in $\bx$ uniformly in $t$,  produces a closed-loop system:
\begin{equation}\label{eq:closed-loop-dyn}
\dot{\bx} = \bf_{\mathrm{cl}}(\bx, t) \coloneqq \bf(\bx, t) + \bG(\bx, t)\bk(\bx, t),
\end{equation}
that is also locally Lipschitz in $\bx$ uniformly in $t$.
Hence, for any initial condition $\bx_0$ and time $t_0$, there exists a maximal interval of existence $I(\bx_0):=[t_0,\tau_{\max})$ such that $t\mapsto\bx(t)$ is the unique solution to \eqref{eq:closed-loop-dyn} on $I(\bx_0)$; cf. \cite[Ch.~3]{Khalil}.

In the CBF methodology, safety is formally established through forward set invariance \cite{AmesTAC17,KondaLCSS21}. A set $\mathcal{C}(t)$ is forward invariant for \eqref{eq:closed-loop-dyn} if, for any initial condition $\bx_0\in\mathcal{C}(t_0)$, the resulting trajectory satisfies $\bx(t)\in\mathcal{C}(t)$ for all $t\in I(\bx_0)$. When using CBFs, safe sets are defined as the zero superlevel set of a continuously differentiable function $h : \R^n\times\R\rightarrow\R$:
\begin{equation}
\begin{aligned}
\mathcal{C}(t) \coloneqq & \{\bx \in \mathbb{R}^n \mid h(\bx, t) \ge 0\}.
\label{eq:C}
\end{aligned}
\end{equation}
CBFs provide conditions for synthesizing feedback controllers that satisfy safety properties.
% Specifically, by providing conditions under which there exists a control input satisfying the barrier condition.
\begin{definition}[\cite{KondaLCSS21}]
    A continuously differentiable function $h:\R^n\times \R\rightarrow\R$ defining $\mathcal{C}$ as in \eqref{eq:C} is said to be a CBF for \eqref{eq:general_knowen_nonlinear_system} if there exists\footnote{Here, $\alpha : \mathbb{R} \to \mathbb{R}$ is an extended class $\mathcal{K}_\infty^\text{e}$ function, i.e., $\alpha$ is continuous, strictly increasing, $\alpha(0)=0$, and $\lim_{r\to\pm\infty}\alpha(r)=\pm\infty$.} $\alpha\in\mathcal{K}_\infty^\text{e}$ such that:
    \begin{align}
        \sup_{\bu\in\R^m} \dot{h}(\bx,\bu, t)
        \geq -\alpha\big(h(\bx, t)\big),\quad \forall (\bx,t)\in\R^n\times\R,
        \label{eq:cbf}
    \end{align}
    where:
    \begin{equation*}
        \dot{h}(\bx,\bu, t) = \tpdv{h}{t}(\bx, t) + \tpdv{h}{\bx}(\bx, t)\left[\bf(\bx,t) + \bG(\bx,t)\bu \right].
    \end{equation*}
\end{definition}
The main utility of CBFs is their ability to instantiate \emph{safety filters}, which minimally modify a nominal controller to enforce safety, typically via quadratic programming \cite{AmesTAC17}. 

\subsection{Problem Formulation}
The main objectives of this paper are twofold. First, given a smooth time-varying vector field $\bv_{\rm{d}}\,:\,\R^3\times\R\rightarrow\R^3$ that encodes a desired behavior for a quadrotor's center of mass $\dot{\bp} = \bv_{\rm{d}}(\bp, t)$, design a feedback controller $\bu = \bk(\bx, t)$ such that along closed-loop trajectories, $\bv(t)$ exponentially converges to $\bv_{\rm{d}}(\bp(t), t)$.
Second, given a safety constraint on the quadrotor's position, derive conditions on the vector field $\bv_{\rm{d}}$ and feedback controller $\bk$ such that completion of Objective 1 leads to safety of the closed-loop system.
% The main objective is to design a controller such that the quadrotor's center-of-mass velocity $\bv$ tracks a given smooth desired vector field $\bv_{\mathrm{d}}(\bp)$ with a prescribed exponential convergence rate. By ensuring exponential convergence to the desired vector field, the properties encoded in the desired velocity profile can be transferred to the full nonlinear quadrotor dynamics, which may include position tracking or safety specifications depending on the design objective.

\section{Vector Field Tracking}
In this section, we present two approaches to driving the quadrotor's center-of-mass velocity $\bv$ exponentially to a desired vector field $\bv_{\mathrm{d}}(\bp, t)$. Our first approach leverages dynamic extension, and the second approach employs geometric properties of the reduced-attitude representation. 

\begin{remark}
   For brevity, our designs are carried out up to the angular velocity level. Generalizing our designs to the moment level (i.e., with $\bM$ as the input) requires one additional backstepping step.
\end{remark}

\subsection{Dynamic Extension}
Dynamic extension is a common technique in quadrotor control that allows for delaying the appearance of the translational subsystem input (thrust), so that it enters the dynamics at the same level as the rotational subsystem. This endows the position with a uniform relative degree with respect to both the translational and rotation inputs. 
% giving a uniform relative degree to the system, and an invertible input map. 
We dynamically extend \eqref{eq:reduced_model_omega} by appending the thrust $T$ to the original state $\bx$, creating the dynamically extended state $\bxi \coloneqq (\bp, \bv,\bR_3,T)$ and control input $\bnu \coloneqq (T_{e}, \omega_1,\omega_2)$, where $T_{e}\coloneqq \dot{T}$.
Thus, the dynamically extended system is given by:
\begin{equation}\label{eq:de-dyn}
    \dot{\bxi} = 
    \begin{bmatrix}
        \dot{\bp} \\ \dot{\bv} \\ \dot{\bR}_3 \\ \dot{T}
    \end{bmatrix}
    =
    \begin{bmatrix}
        \bv \\ \frac{1}{m}\bR_{3}T + \bg \\ -\bR_{2}(t)\omega_1 + \bR_1(t)\omega_2 \\ T_{e}
    \end{bmatrix},
\end{equation}
which, as we demonstrate below, is in a form amenable to backstepping. To this end, define the error coordinates:
\begin{equation}
\label{eq:erorrs_de}
    \begin{aligned}
        \be_1 \coloneqq \bv - \bv_{\rm{d}}(\bp, t), \qquad
        \be_2 \coloneqq \tfrac{1}{m}\bR_{3}T - \ba_{\rm{d}}(\bp,\bv,t),
    \end{aligned}
\end{equation}
where $\bv_{\rm{d}}(\bp,t)$ is a smooth desired vector field and:
\begin{equation}\label{eq:acceleration-profile}
    \ba_{\rm{d}}(\bp, \bv, t) \coloneqq \dot{\bv}_{\rm{d}}(\bp, t) - \bg - \tfrac{k_1}{2}(\bv - \bv_{\rm{d}}(\bp,t)),\; k_1 > 0,
\end{equation}
is a desired acceleration profile. The error dynamics are:
% which have the dynamics:
\begin{equation}\label{eq:de_error_dyn}
    \begin{bmatrix}
        \dot{\be}_1 \\ \dot{\be}_2
    \end{bmatrix}
    =
    \begin{bmatrix}
        \frac{1}{m}\bR_3T + \bg - \dot{\bv}_{\rm{d}}(\bp, t) \\ 
        \bA(\bxi, t)\bB(\bxi)\bnu - \dot{\ba}_{\rm{d}}(\bp, \bv, t)
    \end{bmatrix},
\end{equation}
where:
\begin{equation}
    \bA(\bxi, t) \coloneqq \begin{bmatrix}
        \bR_{3} & -\bR_{2}(t) & \bR_{1}(t)
    \end{bmatrix},
\end{equation}
is an orthogonal matrix and:
% $\bA(\bxi,t)\in\mathrm{SO}(3)$,0
\begin{equation}
    \bB(\bxi) \coloneqq \mathrm{diag}\left(\frac{1}{m}, \frac{T}{m},\frac{T}{m}\right).
\end{equation}
Now, consider the controller:
\begin{equation}\label{eq:de-tracking-controller}
    \begin{aligned}
        \bnu = \bB(\bxi)^{-1}\bA(\bxi,t)\T\Big[ & \dot{\ba}_{\rm{d}}(\bxi,t) - \sigma\be_1 - \tfrac{k_2}{2}\be_2\Big],
    \end{aligned}
\end{equation}
where $\sigma,k_2>0$. 
% and:
% \begin{equation}
%     \ba_{\rm{d}}(\bp, \bv, t) \coloneqq \dot{\bv}_{\rm{d}}(\bp, t) - \bg - \tfrac{\lambda_1}{2}(\bv - \bv_{\rm{d}}).
% \end{equation}
% The interpretation of $\ba_{\rm{d}}$ is that of a desired acceleration profile.
Note that $\bnu$ in \eqref{eq:de-tracking-controller} is well-defined if $T>0$. To this end, define the set:
\begin{equation}\label{eq:de-state-space}
    \mathcal{M} \coloneqq \R^3\times\R^3\times S^2 \times \R_{>0},
\end{equation}
as the state space of \eqref{eq:de-dyn}, noting that $\mathcal{M}$ is a smooth manifold since $S^2$ is a smooth manifold and $\mathcal{M}$ is defined as a finite product of $S^2$ with open sets in Euclidean space. The closed-loop dynamics, consisting of \eqref{eq:de-dyn} with controller \eqref{eq:de-tracking-controller}, are well-defined and smooth on $\mathcal{M}$ so long as $\bv_{\rm{d}}$ is smooth. Thus, it follows from \cite[Thm~9.12]{LeeSmoothManifolds} that for each initial condition $\bxi_0\in\mathcal{M}$ there exists a unique solution $\bxi\,:\,I(\bxi_0)\rightarrow\mathcal{M}$ to the closed-loop system, defined on the maximal interval of existence $I(\bxi_0)=[t_0, \tau_{\max})$. The following theorem illustrates that \eqref{eq:de-tracking-controller} enforces exponential decay of the tracking errors \eqref{eq:erorrs_de} along trajectories.
% %
% We make the following common assumption when working with dynamically extended quadrotor dynamics.
% %
% \begin{assumption}\label{assumption:positive-thrust}
%     Given the controller in \eqref{eq:de-tracking-controller}, the trajectories of the closed-loop system \eqref{eq:de-dyn} satisfy:
%     \begin{equation*}
%         T(t_0) > 0 \implies T(t) > 0,\quad \forall t\in I(\bxi_0),
%     \end{equation*}
%     where $I(\bxi_0)=[t_0, \tau_{\max})$ denotes the maximal interval of existence of the closed-loop system \eqref{eq:de-dyn}.
%     % $T(t)> 0$ for all $t\geq t_0$ so long as $T(t_0) > 0$.
% \end{assumption}

% Assumption \ref{assumption:positive-thrust} requires the thrust to be positive for all time. This is challenging to verify a priori, but is arguably not restrictive in practice as zero thrust implies the quadrotor is in free-fall. 
% %
% The following theorem illustrates that \eqref{eq:de-tracking-controller} enforces exponential decay of the tracking errors \eqref{eq:erorrs_de} along closed-loop trajectories.
% exponentially drives \eqref{eq:de-dyn} to a desired vector field $\bv_{\rm{d}}$.

\begin{theorem}
\label{th:de_controller}
Consider system \eqref{eq:de-dyn} with $\bxi\in\mathcal{M}$ as in \eqref{eq:de-state-space} and let $\bv_{\mathrm d}(\bp,t)$ be smooth. 
% If Assumption~\ref{assumption:positive-thrust} holds, 
Then, the controller \eqref{eq:de-tracking-controller}
enforces exponential convergence of the error dynamics \eqref{eq:de_error_dyn} along trajectories in that for any initial condition $\bxi_0\in\mathcal{M}$:
% satisfying $T(t_0) > 0$ :
\begin{equation}\label{eq:de-exp-bounds}
    \|\be(t)\| \leq \max\left\{\sqrt{\sigma},\tfrac{1}{\sqrt{\sigma}} \right\}\|\be(t_0)\|e^{-k(t - t_0)},
\end{equation}
for all $t\in I(\bxi_0)$,
where $\be\coloneqq(\be_1,\be_2)$ and $k\coloneqq \min\{\tfrac{k_1}{2},\tfrac{k_2}{2}\}$.
\end{theorem}

\begin{proof}
Define the Lyapunov function candidate:
\begin{equation}
     \label{eq:de-lyap}
      V = \tfrac{1}{2}\be_1\T\be_1 + \tfrac{1}{2\sigma}\be_2\T\be_2.
\end{equation}
Differentiating \eqref{eq:de-lyap} along the dynamics \eqref{eq:de_error_dyn} yields:
\begin{equation}
\begin{aligned}
    \dot V = &  \be_1\T\dot{\be}_1 +\tfrac{1}{\sigma}\be_2\T\dot\be_2=\be_1\T(\tfrac{1}{m}\bR_3T + \bg - \dot{\bv}_{\rm{d}}(\bp, t))\\&+\tfrac{1}{\sigma}\be_2\T(\bA(\bxi, t)\bB(\bxi)\bnu - \dot{\ba}_{\rm{d}}(\bp, \bv, t)).
\end{aligned}
\end{equation}
Substituting the control law \eqref{eq:de-tracking-controller} then results in:
\begin{equation}
    \dot V =\be_1\T(-\tfrac{k_1}{2}\be_1+\be_2)+\be_2\T(-\be_1-\tfrac{k_2}{2\sigma}\be_2),
\end{equation}
which simplifies to:
\begin{equation}
\label{eq:dot-de-lyap}
    \dot V = -\tfrac{k_1}{2}\be_1\T\be_1 - \tfrac{k_2}{2\sigma}\be_2\T\be_2.
\end{equation}
Using \eqref{eq:de-lyap}, the above can be bounded as:
\begin{equation}
    \dot V \le -2k(\tfrac{1}{2}\|\be_1\|^2 + \tfrac{1}{2\sigma}\|\be_2\|^2) \le -2kV.
\end{equation}
% where $k \coloneqq\text{min}\{\tfrac{\lambda_1}{2},\tfrac{\lambda_2}{2}\}$. 
Finally, on the interval of existence $I(\bxi_0):=[t_0,\tau_{\max})$, applying the comparison lemma \cite[Lem.~3.4]{Khalil} results in:
\begin{equation}
    V(t) \leq V(t_0) e^{-2k(t - t_0)},
\end{equation}
which can be rearranged using \eqref{eq:de-lyap} to yield \eqref{eq:de-exp-bounds}.
% implies:
% \begin{equation}
%   \|(\be_1(t),\be_2(t))\| \le \text{max}\{\sqrt{\sigma},\tfrac{1}{\sqrt{\sigma}}\}\|(\be_1(0),\be_2(0))\|\,e^{-kt},
% \end{equation}
% for all $t \in I(\bx_0)$. Therefore, along the closed-loop trajectories the errors \eqref{eq:erorrs_de} exponentially converge to zero.
\end{proof}

The above theorem establishes exponential decay of the tracking errors \eqref{eq:erorrs_de} along trajectories of the closed-loop system \eqref{eq:de-dyn}, implying exponential convergence of $\bv$ to $\bv_{\rm{d}}(\bp,t)$ so long as solutions are defined. Theorem \ref{th:de_controller} makes no claims that solutions are complete (exist for all times), which would require further assumptions on $\bv_{\rm{d}}$ and the behavior of $T$ along trajectories. As the emphasis in this paper is on forward invariance, which does not require completeness of solutions \cite{AmesTAC17}, we omit such an analysis here in the interest of space.

\subsection{Geometric Controller}
In the previous subsection, we presented a controller based on dynamic extension for vector field tracking.
Although dynamic extension is commonly used in practice, it can also introduce undesirable delays since one must integrate the designed controller to produce the physical inputs that are sent to the plant. In this section, instead of dynamically extending \eqref{eq:reduced_model_omega}, we use the geometric properties of $S^2$ to directly design a vector field tracking controller for \eqref{eq:reduced_model_omega}.

To this end, let $\bv_{\rm{d}}(\bp,t)$ be a desired smooth vector field
% Consider the system \eqref{eq:reduced_model_omega} 
and define the error coordinates:
\begin{equation}\label{eq:geometric-errors}
    \begin{aligned}
        \be_1 \coloneqq \bv-\bv_\mathrm{d}(\bp,t), \qquad
        \be_2 \coloneqq \bR_3 \times \bR_{3\mathrm{d}},
    \end{aligned}
\end{equation}
where $\bR_{3\mathrm{d}} \coloneqq \ba_\mathrm{d}(\bp,\bv,t)/{\|\bm \ba_\mathrm{d}(\bp,\bv,t)\|}\in S^2$, and $\ba_{\rm{d}}$ is from \eqref{eq:acceleration-profile}.
That is, $\bR_{3\rm{d}}$ is the unit vector associated with the desired acceleration profile.
% \begin{equation}
%     \ba_\mathrm{d}(\bp,\bv,t)\coloneqq \dot{\bm \bv}_{\mathrm{d}}-\bg-\tfrac{k_1}{2}(\bv-\bm \bv_{\mathrm{d}}), \quad k_1 >0.
% \end{equation}
The error dynamics are:
\begin{equation}\label{eq:geometric_error_dyn}
    \begin{bmatrix}
        \dot{\be}_1 \\ \dot{\be}_2
    \end{bmatrix}
    =
    \begin{bmatrix}
        \frac{1}{m}\bR_3T + \bg - \dot{\bv}_{\rm{d}}(\bp, t) \\ 
        \dot{\bR}_3 \times \bR_\mathrm{3d} + {\bR}_3 \times \dot{\bR}_\mathrm{3d}
    \end{bmatrix}.
\end{equation}
Now consider the controllers for \eqref{eq:reduced_model_omega}:
   \begin{equation}
   \label{eq:geometric-tracking-controllers}
   \begin{aligned}
   T &=  m\|\ba_{\rm{d}}(\bp, \bv, t)\|,\\
    % T = & m\vert \vert -\tfrac{k_1}{2}(\bv-\bm \bv_{\mathrm{d}}) -\bg +\dot{\bm \bv}_{\mathrm{d}} \vert \vert,\\
    \bm \omega^\mathrm{i}  &=k_2\be_2 + \bm \omega_\mathrm{d} + \frac{\|\ba_{\rm{d}}\|}{\sigma}\left[ \frac{-[\bR_3 + \bR_{3\rm{d}}]_{\times}}{1 + \bR_3\T\bR_{3\rm{d}}}\right] \be_1,
    \end{aligned}
    \end{equation}
where $\bm \omega_\mathrm{d} \coloneqq \bR_{3\mathrm{d}} \times \dot \bR_{3\mathrm{d}}$, and $k_1, k_2,\sigma>0$. Note that the controllers are well-defined provided that $\|\ba_\mathrm{d}\| >  0$ and  
% and the reduced attitude does not reach the antipodal point, where 
$\bR_3 \neq -\bR_{3\rm{d}}$.
To this end, define the set:
% \begin{equation}\label{eq:ge-state-space}
%     \mathcal{N}_{\rho} \coloneqq \{\bx \in\R^3\times\R^3\times S^2 \mid 1 - \bR_3^\top\bR_{3\mathrm d}< \rho \},
% \end{equation}
\begin{equation}\label{eq:ge-state-space}
    \mathcal{N}_{\rho}(t) \!\coloneqq\! \left\{\bx \in\mathcal{M} \mid \tfrac{1}{2\sigma}\|\be_1(\bp, \bv, t)\|^2 +   1 - \bR_3^\top\bR_{3\mathrm d}< \rho \right\},
\end{equation}
for $\rho\in(0,2)$ and $\sigma>0$,
where $\mathcal{M} \coloneqq \R^3\times\R^3\times S^2$.
The existence of solutions for each initial condition in $\mathcal{N}_{\rho}(t_0)$ follows analogously to that in the previous subsection so long as $\|\ba_{\rm{d}}(\bp_0,\bv_0,t_0)\|>0$. Under this assumption, for each initial condition $\bx_0\in\mathcal{N}_{\rho}(t_0)$ there exists a unique maximal solution $\bx\,:\,I(\bx_0)\rightarrow\mathcal{M}$ to the closed-loop system, defined on a maximal interval of existence $I(\bx_0)=[t_0, \tau_{\max})$. The following lemma will help facilitate further analysis.

\begin{lemma}\label{lemma:a-b}
    Let $\ba,\bb\in S^2$ and suppose that $\ba\neq-\bb$. Then:
    % \begin{equation}
        \begin{align}
            \ba - \bb = & \frac{(\ba + \bb) \times (\ba \times \bb)}{1 + \ba\T\bb}, \label{eq:a-b} \\ 
            (\ba - \bb)\T = & \frac{-(\ba \times \bb)\T [\ba + \bb]_{\times}}{1 + \ba\T\bb}.\label{eq:a-b-T}
        \end{align}
    % \end{equation}
\end{lemma}

\begin{proof}
    Consider the vector triple products:
    \begin{equation*}
        \begin{aligned}
            \ba \times (\bb\times \ba) = & \bb(\ba\T\ba) - \ba(\ba\T\bb) = \bb - (\ba\T\bb)\ba \\ 
            \bb\times (\ba \times \bb) = & \ba(\bb\T\bb) - \bb(\ba\T\bb) = \ba - (\ba\T\bb)\bb,
        \end{aligned}
    \end{equation*}
    where the second equality in each line follows from the fact that $\ba,\bb\in S^2$ and hence $\|\ba\|^2=\|\bb\|^2 = 1$. Subtracting the first line from the second line then yields:
    \begin{equation*}
        \begin{aligned}
             \bb\times (\ba \times \bb) - \ba \times (\bb\times \ba) = & \ba - \bb + (\ba\T\bb)(\ba - \bb) \\
             = & (1 + \ba\T\bb)(\ba - \bb).
        \end{aligned}
    \end{equation*}
    After noting that $\ba\times (\bb\times \ba) = - \ba\times (\ba \times \bb)$, we have:
    \begin{equation*}
        \begin{aligned}
            (\ba + \bb) \times (\ba \times \bb) =  (1 + \ba\T\bb)(\ba - \bb),
        \end{aligned}
    \end{equation*}
    which, after dividing through by $1 + \ba\T\bb$ yields \eqref{eq:a-b}. Now, taking the transpose of \eqref{eq:a-b} yields:
    \begin{equation*}
        \begin{aligned}
            (\ba - \bb)\T = & \frac{(\ba \times \bb)\T [\ba + \bb]_{\times}\T}{1 + \ba\T\bb} 
            = \frac{-(\ba \times \bb)\T [\ba + \bb]_{\times}}{1 + \ba\T\bb},
        \end{aligned}
    \end{equation*}
    where the second equality follows from the fact that $[\cdot]_{\times}\T = -[\cdot]_{\times}$, thereby showing \eqref{eq:a-b-T}, as desired.
\end{proof}

The following theorem illustrates that \eqref{eq:geometric-tracking-controllers} enforce exponential decay of the tracking errors \eqref{eq:geometric-errors} along trajectories.

\begin{theorem}\label{th:geometric_controller}
    Let $\bv_{\rm{d}}(\bp,t)$ be a smooth vector field and
    consider system \eqref{eq:reduced_model_omega} 
    with the controller in \eqref{eq:geometric-tracking-controllers}. Provided that $\bx_0=(\bp_0,\bv_0,\bR_{30})\in\mathcal{N}_{\rho}(t_0)$ for $\rho\in(0,2)$, where $\mathcal{N}_{\rho}$ is defined as in \eqref{eq:ge-state-space}, and $\|\ba_{\rm{d}}(\bp_0,\bv_0,t_0)\| > 0$,
    then along closed-loop trajectories, the error coordinates \eqref{eq:geometric-errors} satisfy:
    % with $\bx \in \mathcal{N}_\rho$ as in \eqref{eq:ge-state-space} and let
    % $\bv_{\mathrm d}(\bp,t)$ be smooth. Then, the controllers \eqref{eq:geometric-tracking-controllers} enforce exponential convergence of the error dynamics \eqref{eq:geometric_error_dyn} along trajectories in that for any initial condition $\bx_0 \in \mathcal{N}_\rho$, and $\rho\in(0,2)$:
    % \begin{equation}\label{eq:exp_bound}
    %     \|\be_1(t)\|^2+\|\be_2(t)\|^2
    %     \;\le\;
    %     \frac{2}{2-\rho}\Big(\|\be_1(t_0)\|^2+\|\be_2(t_0)\|^2\Big)\,
    %     e^{-k (t-t_0)},
    % \end{equation}
    \begin{equation}\label{eq:exp_bound}
        \|\be(t)\|\leq
         \sqrt{\tfrac{c_2}{c_1}}\|\be(t_0)\|
        e^{-\frac{k}{2} (t-t_0)},\quad \forall t\in I(\bx_0),
    \end{equation}
    where $\be\coloneqq (\be_1,\be_2)$,
    $
    k \coloneqq \min\{k_1,(2-\rho)k_2\} \;>\;0
    $, $c_1 := \min\left\{\tfrac{1}{2\sigma}, \tfrac{1}{2}\right\}$, and
$c_2 := \max\left\{\tfrac{1}{2\sigma}, \tfrac{1}{2-\rho}\right\}$.
\end{theorem}

\begin{proof}
Consider the Lyapunov function candidate:
\begin{equation}
\label{eq:geometric_lyap}
    V(\bx, t) = \underbrace{\tfrac{1}{2\sigma}\|\be_1(\bp, \bv, t)\|^2}_{V_1} + \underbrace{1 - \bR_3\T\bR_{3\rm{d}}}_{V_2}.
\end{equation}
Computing the time derivative of $V_1$ yields:
\begin{equation}
\begin{aligned}
    % V_1 &=  \frac{1}{2}(\bv-\bm \bv_{\mathrm{d}})^T (\bv-\bm \bv_{\mathrm{d}}),\\
    \dot V_1 = &  \tfrac{1}{\sigma}\be_1\T(\tfrac{1}{m}\bR_3T + \bg - \dot {\bm \bv}_{\mathrm{d}}) \\
    = & \tfrac{1}{\sigma}\be_1\T(\tfrac{1}{m}\bR_3T + \bg - \dot {\bm \bv}_{\mathrm{d}} \pm \ba_{\rm{d}}) \\
    = & -\tfrac{k_1}{2\sigma}\|\be_1\|^2 + \tfrac{1}{\sigma}\be_1\T(\tfrac{1}{m}\bR_3T - \ba_{\rm{d}}), \\
\end{aligned}
\end{equation}
where the second equality follows from adding a zero and the third from \eqref{eq:acceleration-profile}. Now, substituting in $T$ from \eqref{eq:geometric-tracking-controllers} yields:
\begin{equation}
    \begin{aligned}
    \label{eq:dot-v1-ge}
        \dot V_1 = & -\tfrac{k_1}{2\sigma}\|\be_1\|^2 + \tfrac{1}{\sigma}\be_1\T(\bR_3\|\ba_{\rm{d}}\| - \ba_{\rm{d}}) \\
        = &-\tfrac{k_1}{2\sigma}\|\be_1\|^2 + \tfrac{1}{\sigma}\be_1\T(\bR_3\|\ba_{\rm{d}}\| - \bR_{3\rm{d}}\|\ba_{\rm{d}}\|) \\
        = & -\tfrac{k_1}{2\sigma}\|\be_1\|^2 + \tfrac{1}{\sigma}\|\ba_{\rm{d}}\|\be_1\T(\bR_3 - \bR_{3\rm{d}}) \\
        = & -\tfrac{k_1}{2\sigma}\|\be_1\|^2 + (\bR_3 - \bR_{3\rm{d}})\T \be_1 \tfrac{\|\ba_{\rm{d}}\|}{\sigma},
    \end{aligned}
\end{equation}
where the second equality follows from the definition of $\bR_{3\rm{d}}$. Next, computing the time derivative of $V_2$ yields:
\begin{equation}
    \begin{aligned}
        \dot{V}_2 = &  -\dot{\bR}_3\T\bR_{3\rm{d}} - \bR_3\T\dot{\bR}_{3\rm{d}} \\
        = & -(\bomega^\mathrm{i} \times \bR_3)\T\bR_{3\rm{d}} - \bR_3\T(\bomega_{\rm{d}} \times\bR_{3\rm{d}}) \\
        = & -\bomega^{\mathrm{i}\top} (\bR_3 \times \bR_{3\rm{d}}) - \bomega_{\rm{d}}\T(\bR_{3\rm{d}} \times \bR_3) \\
        = & -\bomega^{\mathrm{i}\top}(\bR_3 \times \bR_{3\rm{d}}) + \bomega_{\rm{d}}\T(\bR_3 \times \bR_{3\rm{d}}) \\
        = & \be_2\T(\bomega_{\rm{d}} - \bomega^\mathrm{i} ).
    \end{aligned}
\end{equation}
where we used the fact that $\dot{\bR}_3 = \bomega^\mathrm{i}  \times \bR_{3}$. Therefore:
\begin{equation*}
\label{eq:dot-v-ge-cross}
    \begin{aligned}
        \dot{V} = & -\tfrac{k_1}{2\sigma}\|\be_1\|^2 + (\bR_3 - \bR_{3\rm{d}})\T \be_1 \tfrac{\|\ba_{\rm{d}}\|}{\sigma} + \be_2\T(\bomega_{\rm{d}} - \bomega_\mathrm{i} ).
    \end{aligned}
\end{equation*}
It follows from Lemma \ref{lemma:a-b} that:
\begin{equation*}
    \begin{aligned}
        (\bR_3 - \bR_{3\rm{d}})\T = & \tfrac{-(\bR_3 \times \bR_{3\rm{d}})\T[\bR_3 + \bR_{3\rm{d}}]_{\times}}{1 + \bR_3\T\bR_{3\rm{d}}} \\
        = & \be_2\T\left[ \tfrac{-[\bR_3 + \bR_{3\rm{d}}]_{\times}}{1 + \bR_3\T\bR_{3\rm{d}}}\right],
    \end{aligned}
\end{equation*}
which, substituting back into the derivative of $V$, gives:
\begin{equation*}
    \begin{aligned}
        \dot{V} = & -\tfrac{k_1}{2\sigma}\|\be_1\|^2 + \tfrac{\|\ba_{\rm{d}}\|}{\sigma}\be_2\T\left[ \tfrac{-[\bR_3 + \bR_{3\rm{d}}]_{\times}}{1 + \bR_3\T\bR_{3\rm{d}}}\right] \be_1  \\&+ \be_2\T(\bomega_{\rm{d}} - \bomega^\mathrm{i} ) \\
        = & -\tfrac{k_1}{2\sigma}\|\be_1\|^2 \\ & + \be_2\T\left(\bomega_{\rm{d}} - \bomega^\mathrm{i}  +  \tfrac{\|\ba_{\rm{d}}\|}{\sigma} \left[ \tfrac{-[\bR_3 + \bR_{3\rm{d}}]_{\times}}{1 + \bR_3\T\bR_{3\rm{d}}}\right] \be_1\right).
    \end{aligned}
\end{equation*}
Designing $\bomega^\mathrm{i} $ as in \eqref{eq:geometric-tracking-controllers} thus results in:
\begin{equation}
\label{eq:vdot_exact}
    \dot V =  -\tfrac{k_1}{2\sigma}\|\be_1\|^2 -k_2\|\be_2\|^2.
\end{equation}
Now, note that $\|\be_2\|^2 = \|\bR_{3} \times \bR_{3\rm{d}}\|^2 = \sin^2\theta = 1 - \cos^2\theta,$
% \begin{equation*}
%     \|\be_2\|^2 = \|\bR_{3} \times \bR_{3\rm{d}}\|^2 = \sin^2\theta = 1 - \cos^2\theta,
% \end{equation*}
where $\theta\in S^1$ is the angle between $\bR_{3}$ and $\bR_{3\rm{d}}$. Hence:
\begin{equation*}
    \begin{aligned}
        \|\be_2\|^2 = 1 - \cos^2\theta = & (1 - \cos\theta)(1 + \cos\theta) \\
        = & (1 - \bR_3\T\bR_{3\rm{d}})(1 + \bR_3\T\bR_{3\rm{d}}) \\
        = & V_2(2 - V_2),
    \end{aligned}
\end{equation*}
allowing $\dot{V}$ to be expressed as:
\begin{equation}
    \dot{V} = -k_1V_1 - k_2V_2(2 - V_2).
\end{equation}
Now, if $\bx\in\mathcal{N}_{\rho}(t)$, then
$V(\bx, t) < \rho$, which implies that $V_2 < \rho$. It then follows that
% $V_2 < \rho$, which implies that 
$2 - V_2 > 2- \rho$, and hence that:
\begin{equation}
\label{eq:dot-v-geometric}
    \begin{aligned}
        \dot{V}(\bx, t) \leq -k_1V_1 - (2 - \rho)k_2 V_2 \leq -kV(\bx, t),
    \end{aligned}
\end{equation}
for all $\bx\in\mathcal{N}_{\rho}(t)$,
where $k\coloneqq \min\{k_1,(2 - \rho)k_2\}$.
Now, let $\bx_0\in\mathcal{N}_{\rho}(t_0)$ and suppose that $\|\ba_{\rm{d}}(\bp_0,\bv_0,t_0)\|>0$. Since $\rho<2$, it follows from \eqref{eq:ge-state-space} that the closed-loop dynamics \eqref{eq:reduced_model_omega} with controller \eqref{eq:geometric-tracking-controllers} are well-defined and locally Lipschitz in an open neighborhood of $\bx_0$. Hence, there exists a unique maximal solution $\bx\,:\,I(\bx_0)\rightarrow\mathcal{M}$ to the closed-loop system defined on a maximal interval of existence $I(\bx_0)$. Along this solution, it follows from \eqref{eq:dot-v-geometric} and the comparison lemma \cite[Lem.~3.4]{Khalil} and \eqref{eq:dot-v-geometric} that $V(\bx(t), t) \leq e^{-k(t - t_0)}V(\bx_0,t_0)\leq V(\bx_0,t_0)$ for all $t\in I(\bx_0)$. Since $\bx_0\in\mathcal{N}_{\rho}(t_0)$, we have $V(\bx_0,t_0) < \rho$ and that $V_2(t) < V(\bx(t), t) < \rho$ for all $t\in I(\bx_0)$, which implies that $\bx(t)\in\mathcal{N}_{\rho}(t)$ and $\bR_3(t)\neq -\bR_{3\rm{d}}(\bp(t),\bv(t),t)$ for all $t\in I(\bx_0)$.
% Since $\rho<2$
% strictly, suppose a trajectory starts in $\mathcal{N}_\rho$, i.e.\ $V_2(0)<\rho$, and $\|\ba_\mathrm{d}(\bp_0,\bv_0,t_0)\|>0$.
% As long as the trajectory remains in $\mathcal{N}_\rho$, the vector field
% is well-defined there and \eqref{eq:vdot_exact} gives $\dot V\le0$;
% combined with $V_2\le V$ (from $\tfrac12\|\be_1\|^2\ge0$), this yields
% $V_2\le V(t)\le V(t_0)<\rho$, so the trajectory cannot leave
% $\mathcal{N}_\rho$ except possibly by approaching its boundary
% $V_2=\rho$ from within — which the strict inequality $V_2<\rho$
% rules out. Hence $\mathcal{N}_\rho$ is forward invariant, the trajectory
% never approaches the singularity, and the upper bound in
% \eqref{eq:dot-v-geometric} holds along the entire trajectory. Finally, on the interval of existence $I(\bx_0):=[t_0,\tau_{\max})$, applying the comparison lemma \cite[Lem.~3.4]{Khalil} results in:
% \begin{equation}
% V(t) \le V(t_0)\,e^{-k (t - t_0)}.
% \end{equation}
Since $\|\be_2\|^2 = V_2(2-V_2)$ with $V_2\in[0,2]$, we have $\|\be_2\|^2 \le 2V_2$, while on $\mathcal{N}_\rho$ the bound $2-V_2 > 2-\rho$ gives $\|\be_2\|^2 \ge (2-\rho)V_2$. Hence $\tfrac{1}{2}\|\be_2\|^2 \le V_2 \le \tfrac{1}{2-\rho}\|\be_2\|^2$, and since $\tfrac{1}{2} \le \tfrac{1}{2-\rho}$ for $\rho\in(0,2)$:
\begin{equation}
\label{eq:sandwich}
    \tfrac{1}{2}\left(\tfrac{1}{\sigma}\|\be_1\|^2 + \|\be_2\|^2\right) \le V \le \left(\tfrac{1}{2\sigma}\|\be_1\|^2 + \tfrac{1}{2-\rho}\|\be_2\|^2\right).
\end{equation}
Combining \eqref{eq:sandwich}
with the bounds on $V$ obtained via the comparison lemma
% with the lower bound at time $t$ and the upper bound at time $t_0$, 
yields \eqref{eq:exp_bound}, as desired.
% Applying the bounds \eqref{eq:sandwich} once more (lower bound at
% time $t$, upper bound at time $0$ gives \eqref{eq:exp_bound}. Since $\rho\in(0,2)$
% was arbitrary and $\bigcup_\rho \mathcal{N}_\rho$ covers $S^2$ except the
% antipodal point $\bR_3^\top\bR_{3\mathrm d}=-1$, the result holds almost
% globally.
\end{proof}

Similar remarks regarding completeness of solutions made after Theorem \ref{th:de_controller} also apply in the context of Theorem \ref{th:geometric_controller}.

\section{Safety-Critical Control}
In the previous section, we designed controllers enforcing exponential convergence of $\bv$ to $\bv_{\rm{d}}(\bp,t)$. In this section, we illustrate how
these previous results may be used to
enforce safety constraints on a quadrotor's position. 
Here, we fix the \emph{structure} of the controllers developed in the previous section, and showcase how safety can be enforced through the design of the desired vector field. To this end, we consider:
% Our objective is to  encode safety objectives into the desired vector field for the dynamically extended system \eqref{eq:de-dyn}, and \eqref{eq:reduced_model_omega},  ensuring that the quadrotor's center of mass position remains within a prescribed safe set: 
% This candidate safe set is defined as:
\begin{equation}\label{eq:C-quadrotor}
    \mathcal{C} = \{\bp\in\R^3 \mid h(\bp) \ge 0\},
\end{equation}
where $h:\R^3\rightarrow\R$ is smooth, as the constraint set. Our objective is to ensure that $h(\bp(t))\geq0$ along closed-loop trajectories generated by the dynamic extension controller \eqref{eq:de-tracking-controller} or geometric controller \eqref{eq:geometric-tracking-controllers} when tracking $\bv_{\rm{d}}$.
To do so, we design $\bv_{\mathrm{d}}$ to satisfy the robust CBF condition:
% \MC{Replace ISSf term with bounded gradient assumption.}
\begin{equation}\label{eq:safe-desired-vector field}
    \begin{aligned}
        \tpdv{h}{\bp}(\bp) \bv_{\mathrm{d}}(\bp, t) \geq - \alpha h(\bp)+ \tfrac{1}{\epsilon}\|\tfrac{\partial h}{\partial \bp}(\bp)\|^2,
    \end{aligned}
\end{equation}
for each 
$\bp\in\R^3$,
% $(\bq,\bhat{\theta})\in\R^{n}\times\R^p$, 
where $\alpha>0$ represents a linear extended class $\mathcal{K}_{\infty}$ function, $s\mapsto \alpha s$, as in \eqref{eq:cbf} and $\epsilon>0$ is a robustness parameter, inspired by the input-to-state-safety framework \cite{AmesLCSS19}. Note that a smooth $\bv_{\rm{d}}$ satisfying \eqref{eq:safe-desired-vector field} can be designed using the tools from \cite{CohenLCSS23}. The intuition behind \eqref{eq:safe-desired-vector field} is that \emph{if} $\bv = \bv_{\rm{d}}(\bp,t)$, then \eqref{eq:safe-desired-vector field} implies that $\dot{h} = \pdv{h}{\bp}\bv\geq - \alpha(h)$, which implies that $h$ satisfies the CBF condition \eqref{eq:cbf}. We will thus use the controllers from the previous section to drive $\bv$ to $\bv_{\rm{d}}$ while ensuring safety even during transients.

\subsection{Safety with Dynamic Extension Controller}
We begin by establishing safety properties for the dynamic extension vector field tracking controller from \eqref{eq:de-tracking-controller}. To this end, consider the CBF candidate:
\begin{equation}\label{eq:barrier-de}
    \begin{aligned}
        B(\bxi,t) \coloneqq h(\bp) - \tfrac{1}{\mu}V(\bxi, t),\quad \mu > 0,
    \end{aligned}
\end{equation}
for the dynamically extended system \eqref{eq:de-dyn},
where $V$ is the Lyapunov function from \eqref{eq:de-lyap}, which induces a safe set:
\begin{equation}\label{eq:de-safe-set}
    \mathcal{S}(t) \coloneqq \{\bxi\in\mathcal{M}\mid B(\bxi, t) \geq0\}.
\end{equation}
The following theorem provides conditions under which tracking $\bv_{\rm{d}}$ in \eqref{eq:safe-desired-vector field} using the dynamic extension controller \eqref{eq:de-tracking-controller} implies that $\bxi(t)\in \mathcal{S}(t)$ for all $t\in I(\bxi_0)$.

\begin{theorem}\label{th:safety_robot}
      Let $\bv_\mathrm{d}\,:\,\R^3\times\R\rightarrow\R^3$ be any smooth desired vector field satisfying \eqref{eq:safe-desired-vector field} and let the hypotheses of Theorem \ref{th:de_controller} hold. Provided that:
    \begin{equation}\label{eq:safe-tracking-conditions}
          k_1 \geq \max\left\{\epsilon\mu,~2\alpha\right\}, \quad k_2 \ge \alpha,
    \end{equation}
    then $\mathcal{S}$ as in \eqref{eq:de-safe-set} is forward invariant for the closed-loop system \eqref{eq:de-dyn}, and $\bp(t)\in\mathcal{C}$ for all $t\in I(\bxi_0)$.
\end{theorem}

\begin{proof}
Taking time derivative of $B$ yields:
\begin{equation*}
        \begin{aligned}
            \dot{B}(\bxi,t) = & \tpdv{h}{\bp}(\bp)\bv - \tfrac{1}{\mu}\dot{V}(\bxi,t) \\ 
            = & \tpdv{h}{\bp}(\bp)\bv_\mathrm{d} + \tpdv{h}{\bp}(\bp)\be_1- \tfrac{1}{\mu}\dot{V}(\bxi,t) ,
        \end{aligned}
\end{equation*}
where the second equality follows from $\bv=\bv_{\rm{d}} + \be_1$, where $\be_1$ is defined as in \eqref{eq:erorrs_de}.
Lower bounding the above yields:
\begin{equation*}
    \begin{aligned}
        \dot{B} \geq & -\alpha h  + \tfrac{1}{\epsilon}\|\tpdv{h}{\bp}\|^2 - \|\tpdv{h}{\bp}\|\|\be_1\| +\tfrac{k_1}{2\mu}\|\be_1\|^2
+\tfrac{k_2}{2\sigma\mu}\|\be_2\|^2,
    \end{aligned}
\end{equation*}
where we used \eqref{eq:safe-desired-vector field} and \eqref{eq:dot-de-lyap}. Now, note that:
\begin{equation*}
    \begin{aligned}
        \tfrac{1}{\epsilon}\|\tpdv{h}{\bp}\|^2 - \|\tpdv{h}{\bp}\|\|\be_1\| + \tfrac{k_1}{4\mu}\|\be_1\|^2 \\ 
        = 
        \begin{bmatrix}
          \|\pdv{h}{\bp}\| & \|\be_1\| 
        \end{bmatrix} 
        \begin{bmatrix}
            \frac{1}{\epsilon} & -\frac{1}{2}\\ -\frac{1}{2} & \frac{k_1}{4\mu}
        \end{bmatrix}
        \begin{bmatrix}
           \|\pdv{h}{\bp}\| \\ \|\be_1\| 
        \end{bmatrix}.
    \end{aligned}
\end{equation*}
If $k_1\ge {\epsilon\mu}$, the above matrix is positive semi-definite, allowing $\dot{B}$ to be further lower bounded as:
\begin{equation*}
    \begin{aligned}
        \dot{B} \ge &  -\alpha h + \tfrac{k_1}{4\mu}\|\be_1\|^2 +\tfrac{k_2}{2\sigma\mu}\|\be_2\|^2 \\
        = & -\alpha[B + \tfrac{1}{\mu}V] + \tfrac{k_1}{4\mu}\|\be_1\|^2 +\tfrac{k_2}{2\sigma\mu}\|\be_2\|^2 \\
        = & -\alpha B + \tfrac{1}{2\mu}[\tfrac{k_1}{2} - \alpha]\|\be_1\|^2 + \tfrac{1}{2\sigma\mu}[k_2 - \alpha]\|\be_2\|^2,
    \end{aligned}
\end{equation*} 
where the second line follows from \eqref{eq:barrier-de} and the third from \eqref{eq:de-lyap}. If
$k_1\geq 2\alpha$ and $k_2 \ge \alpha$, we may bound $\dot{B}$ as $\dot{B}(\bxi,t) \ge -\alpha B(\bxi, t)$ for any $\bxi\in\mathcal{M}$ and any $t\in\R$. Now, for any $\bxi_0\in\mathcal{S}$ there exists a unique maximal solution $\bxi\,:\, I(\bxi_0)\rightarrow\mathcal{M}$ defined on a maximal interval of existence $I(\bxi_0)$. It then follows from the comparison lemma \cite[Lem.~3.4]{Khalil} that: 
\begin{equation*}
    B(\bxi(t), t) \geq e^{-\alpha(t - t_0)}B(\bxi_0, t_0) \geq 0,\quad \forall t\in I(\bxi_0),
\end{equation*}
where the second inequality follows from the fact that $B(\bxi_0,t_0)\geq0$. Hence, $\mathcal{S}$ is forward invariant. Further, it follows from \eqref{eq:barrier-de} that for all $t\in I(\bxi_0)$:
\begin{equation*}
     B(\bxi(t), t) \geq 0 \implies h(\bp(t)) \geq \tfrac{1}{\mu}V(\bxi(t),t)\geq 0,
\end{equation*}
implying that $\bp(t)\in\mathcal{C}$ for all $t\in I(\bxi_0)$, as desired.
% \begin{equation}
%     \dot{B}(\bxi,t) \ge -\alpha B(\bxi, t),\quad\forall (\bxi,t)\in\mathcal{M}\times\R,
% \end{equation}
% which implies that ${B}(\bz,t)$ is a barrier function for \eqref{eq:de-dyn}.
% Invoking \cite[Theorem 1]{mirtaba2026high} then implies that the set $\mathcal{S}$ from \eqref{eq:S-adaptive} is forward invariant for the dynamics in \eqref{eq:de-dyn}.
\end{proof}

\subsection{Safety with Geometric Controller}
We now shift to establishing safety of \eqref{eq:reduced_model_omega} when using the geometric controller \eqref{eq:geometric-tracking-controllers} to track a smooth vector field satisfying \eqref{eq:safe-desired-vector field}. Here, we consider the CBF candidate:
\begin{equation}\label{eq:barrier-geometric}
    \begin{aligned}
        B(\bx,t) \coloneqq h(\bp) - \tfrac{1}{\mu}V(\bx, t),\quad \mu > 0,
    \end{aligned}
\end{equation}
where $h$ is from \eqref{eq:C-quadrotor} and 
where $V$ is the Lyapunov function from \eqref{eq:geometric_lyap}. This CBF candidate induces a safe set:
\begin{equation}\label{eq:ge-safe-set}
    \mathcal{S}(t) \coloneqq \{\bx\in\mathcal{N}_\rho\mid B(\bx, t) \geq0\}.
\end{equation}
The following theorem provides conditions under which tracking $\bv_{\rm{d}}$ in \eqref{eq:safe-desired-vector field} using the geometric controllers \eqref{eq:geometric-tracking-controllers} implies that $\bx(t)\in \mathcal{S}(t)$ for all $t\in I(\bx_0)$. 

\begin{theorem}\label{th:safety_geometric_robot}
      Let $\bv_\mathrm{d}(\bp,t)$ be a smooth desired vector field satisfying \eqref{eq:safe-desired-vector field} and let the hypotheses of Theorem 2 hold for the dynamics \eqref{eq:reduced_model_omega} and controller \eqref{eq:geometric-tracking-controllers}. Provided that:
    \begin{equation}\label{eq:safe-tracking-conditions_geometric}
          k_1 \geq \max\left\{\epsilon\sigma\mu,~2\alpha\right\}, \quad k_2 \ge \tfrac{\alpha}{2-\rho},
    \end{equation}
    $\bx_0 \in \mathcal{S}$, for $\mathcal{S}$ as in \eqref{eq:ge-safe-set}, and $\bx_0$ satisfies the conditions of Theorem 2, then along trajectories of the closed-loop system, we have $\bp(t)\in\mathcal{C}$ for all $t\in I(\bx_0)$.
\end{theorem}

\begin{proof}
Taking time derivative of $B$ yields:
\begin{equation*}
        \begin{aligned}
            \dot{B}(\bx,t) = & \tpdv{h}{\bp}(\bp)\bv - \tfrac{1}{\mu}\dot{V}(\bx,t) \\ 
            = & \tpdv{h}{\bp}(\bp)\bv_\mathrm{d} + \tpdv{h}{\bp}(\bp)\be_1- \tfrac{1}{\mu}\dot{V}(\bx,t).
        \end{aligned}
\end{equation*}
Lower bounding the above yields:
\begin{equation*}
    \begin{aligned}
        \dot{B} \!\geq\! & -\alpha h  \!+\! \tfrac{1}{\epsilon}\|\tpdv{h}{\bp}\|^2 - \|\tpdv{h}{\bp}\|\|\be_1\| \!+\! \tfrac{k_1}{2\sigma\mu}\|\be_1\|^2
        +\tfrac{k_2}{\mu}\|\be_2\|^2,
    \end{aligned}
\end{equation*}
where we used \eqref{eq:safe-desired-vector field} and \eqref{eq:vdot_exact}. Now, note that:
\begin{equation*}
    \begin{aligned}
        \tfrac{1}{\epsilon}\|\tpdv{h}{\bp}\|^2 - \|\tpdv{h}{\bp}\|\|\be_1\| + \tfrac{k_1}{4\sigma\mu}\|\be_1\|^2 \\ 
        = 
        \begin{bmatrix}
          \|\pdv{h}{\bp}\| & \|\be_1\| 
        \end{bmatrix} 
        \begin{bmatrix}
            \frac{1}{\epsilon} & -\frac{1}{2}\\ -\frac{1}{2} & \frac{k_1}{4\sigma\mu}
        \end{bmatrix}
        \begin{bmatrix}
           \|\pdv{h}{\bp}\| \\ \|\be_1\| 
        \end{bmatrix}.
    \end{aligned}
\end{equation*}
If $k_1\ge \epsilon\sigma\mu$, the above matrix is positive semi-definite, allowing $\dot{B}$ to be further lower bounded as:
% \begin{equation*}
%     \dot{B}(\bx,t) \ge -\alpha h + \tfrac{k_1}{4\mu}\|\be_1\|^2 +\tfrac{k_2}{\mu}\|\be_2\|^2.
% \end{equation*} Now, using the definition of $B$ from \eqref{eq:B}, and the Lyapunov function \eqref{eq:geometric_lyap} we have:
% % Using the definition of $B(\bz, t)$ and lower bounding it gives:
% \begin{equation*}
%     B(\bz, t) \ge h(\bp)-\tfrac{1}{2\mu}\|\bv-\bv_{\mathrm{d}}\|^2-
%     \tfrac{1}{\mu}V_2,
% \end{equation*}
% which allows $\dot{B}$ to be further bounded as:
% which leads to:
% \begin{align}
%      -\alpha[h(\bx)-\frac{1}{2}\btilde{\vartheta}\T\bm\Gamma^{-1}\btilde{\vartheta}] \ge -\alpha[ B(\bz, t)+\frac{\overline{M}}{2\mu}\|\bs\|^2+\frac{1}{2\mu}\btilde{\nu}\T\bm\Gamma^{-1}\btilde{\nu}],
% \end{align}
% resulting in:
\begin{equation*}
    \begin{aligned}
        \dot{B} \geq & -\alpha h  + \tfrac{k_1}{4\sigma\mu}\|\be_1\|^2+\tfrac{k_2}{\mu}\|\be_2\|^2 \\
        \geq & -\alpha B + \tfrac{1}{2\sigma\mu}[\tfrac{k_1}{2} - \alpha]\|\be_1\|^2 -\tfrac{\alpha}{\mu} V_2 + \tfrac{k_2}{\mu}\|\be_2\|^2  \\
        \geq & -\alpha B + \tfrac{1}{2\sigma\mu}[\tfrac{k_1}{2} - \alpha]\|\be_1\|^2 + \tfrac{1}{\mu}(k_2-\tfrac{\alpha}{2-\rho} )\|\be_2\|^2,
    \end{aligned}
\end{equation*}
where the third inequality follows from the conclusions of Theorem \ref{th:geometric_controller} and the fact that $V_2 \le \tfrac{1}{2-\rho}\|\be_2\|^2$ for $\bx\in\mathcal{N}_{\rho}$. Thus, if $k_1\geq 2\alpha$, and $k_2 \ge \tfrac{\alpha}{2-\rho}$ hold, we may bound $\dot{B}$ as:
\begin{equation}\label{eq:B-dot-geometric-safety}
    \dot{B}(\bx,t) \ge -\alpha B(\bx, t),\quad\forall (\bx,t)\in 
    % \mathcal{M} 
    \mathcal{N}_{\rho}(t)
    \times \mathbb{R},
\end{equation}
If $\bx_0\in\mathcal{N}_{\rho}(t_0)$ it follows from the proof of Thm. \ref{th:geometric_controller} that solutions exist on a maximal interval of existence and satisfy $\bx(t)\in\mathcal{N}_{\rho}(t)$ for $t\in I(\bx_0)$. Thus, \eqref{eq:B-dot-geometric-safety} holds along solutions. The remainder of the proof follows that of Thm. \ref{th:safety_robot}.
% Similar remarks regarding existence of a unique maximal solution and applying comparison lemma to imply $ \bx(t) \in \mathcal{S}, \forall t\in I(\bx_0)$ for system \eqref{eq:reduced_model_omega} also apply here.
\end{proof}

\begin{remark}
The intuition behind the sufficient gain conditions in Theorems \ref{th:safety_robot} and \ref{th:safety_geometric_robot} is: 1) initialization within the safe set and where the dynamics are well-defined; 2) vector field tracking should be sufficiently fast relative to the variations in the safe vector field, captured by $\alpha$.
\end{remark}

% \noindent Moreover,
% $
% \dot V_2=0
% \iff
% \be_1=\mathbf0,
% \;
% e_2=0.
% $
% Let
% $
% E
% \coloneqq
% \left\{
% (\be_1,e_2)\in D:
% \be_1=\mathbf0,\;
% e_2=0
% \right\},
% $
% which is the largest invariant subset of
% $
% \left\{
% (\be_1,e_2)\in D:
% \dot V_2=0
% \right\}.
% $
% Therefore, by LaSalle's invariance principle,
% $
% \be_1(t)\rightarrow\mathbf0
% $
% and
% $
% e_2(t)\rightarrow0
% $
% as
% $
% t\rightarrow\infty.
% $
% Equivalently,
% $
% \bv(t)-\bm\beta_1(t)\rightarrow\mathbf0
% $
% and
% $
% 1-\bR_3^{T}(t)\bR_{3\mathrm d}(t)\rightarrow0.
% $
% Hence,
% $
% \bR_3(t)\rightarrow\bR_{3\mathrm d}(t),
% $
% for all initial conditions satisfying
% $
% \bR_3^{T}(0)\bR_{3\mathrm d}(0)>-1.
% $
% \noindent Therefore, the origin of the closed-loop error dynamics is almost globally exponentially stable.

\section{Numerical Simulations}
In this section, we demonstrate the effectiveness of our vector field tracking controllers via numerical simulations. 
% We first evaluate their reference trajectory tracking performance and then demonstrate constraint enforcement using both the dynamic extension and geometric controllers.

\textbf{Trajectory Tracking.}
We begin by illustrating the ability of our controllers to track vector fields without any safety considerations.
Here, the quadrotor mass is $m = 1 kg$ and the initial condition is: \[
\bp_0 = \bm 0, \quad \bv_0 = \bm 0,\quad \bR_0 =
\begin{bsmallmatrix}
1 & 0 & 0 \\
0 & \cos{\theta_0} & -\sin{\theta_0} \\
0 & \sin{\theta_0} & \cos{\theta_0}
\end{bsmallmatrix},
\]
and $\theta_0 = 175^\circ$, which represents a nearly inverted initial attitude. For the DE controller, the extended state initial condition is $T_0= 10$. The desired vector field is 
$\bv_\mathrm{d}(\bp,t) = -\lambda(\bp -\bp_\mathrm{d}(t)) + \dot{\bp}_\mathrm{d}(t)$, $\lambda=1$, encoding tracking of a reference trajectory $\bp_\mathrm{d}(t) = [\cos{2t} \; \sin{2t} \; 1]\T \mathrm{m}$. The controller gains are $k_1=k_2 = 4$ for the DE controller, and $k_1 =4,k_2=2$ for the geometric controller. Fig.~\ref{fig:de_track} shows the DE controller, and Fig.~\ref{fig:geometric_track} shows the geometric controller results.  
Simulation results show that the velocity errors exponentially converge to zero in accordance with Theorems \ref{th:de_controller} and \ref{th:geometric_controller}, even when initialized with a large rotational error.

\begin{figure}[h]
    \centering
    % \vspace{-2mm}
    \includegraphics[width=0.8\columnwidth]{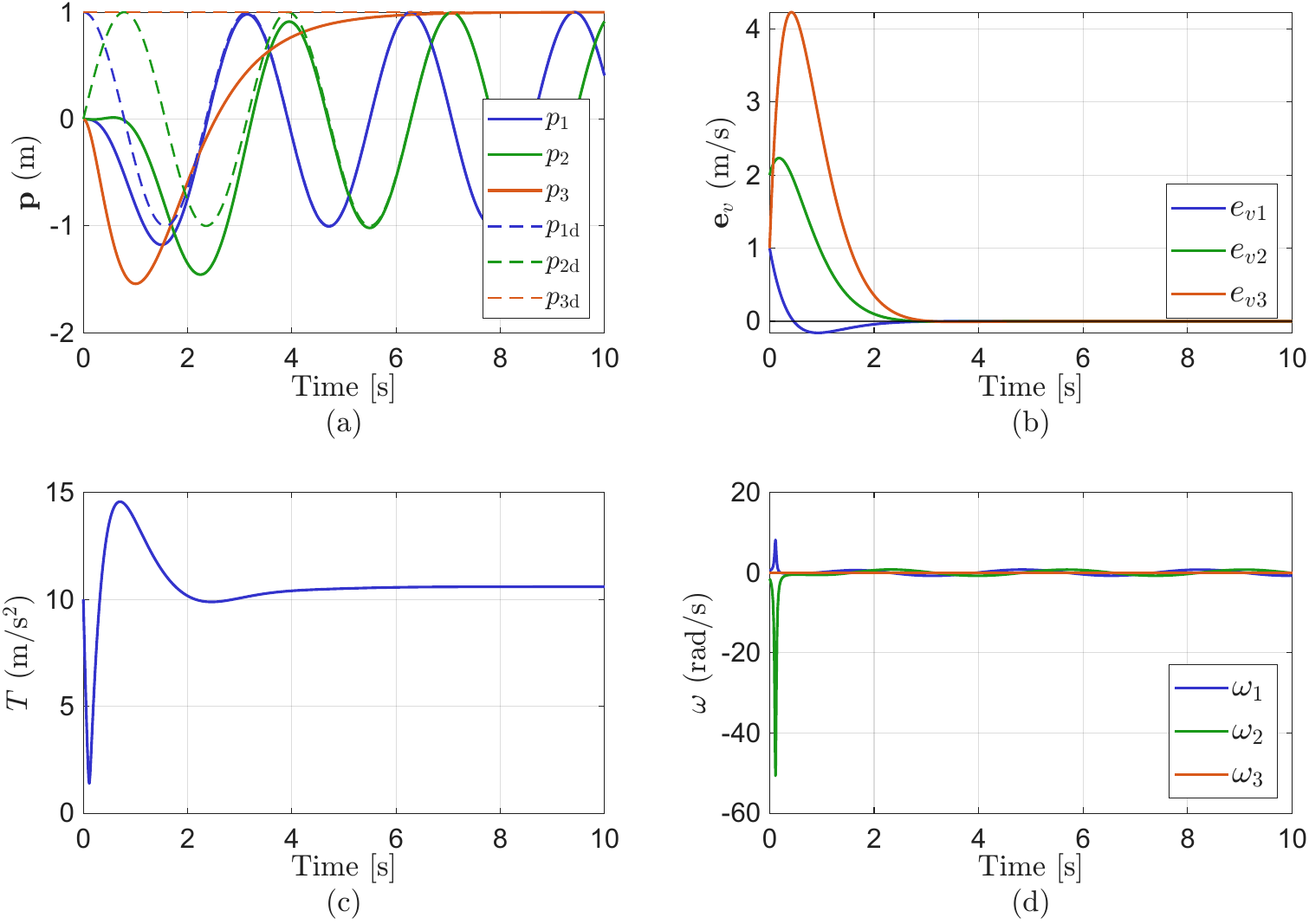}
    \vspace{-4mm}
    \caption{\textbf{DE} controller Tracking. (a) shows quadrotor position $\bp$ (solid), the desired position $\bp_\mathrm{d}$ (dashed). (b) shows velocity errors, (c) shows the total thrust, and (d) shows input angular velocities $\bm \omega$.}
    \label{fig:de_track}
    % \vspace{-2mm}
\end{figure}

\begin{figure}[h]
    \centering
    % \vspace{-2mm}
    \includegraphics[width=0.8\columnwidth]{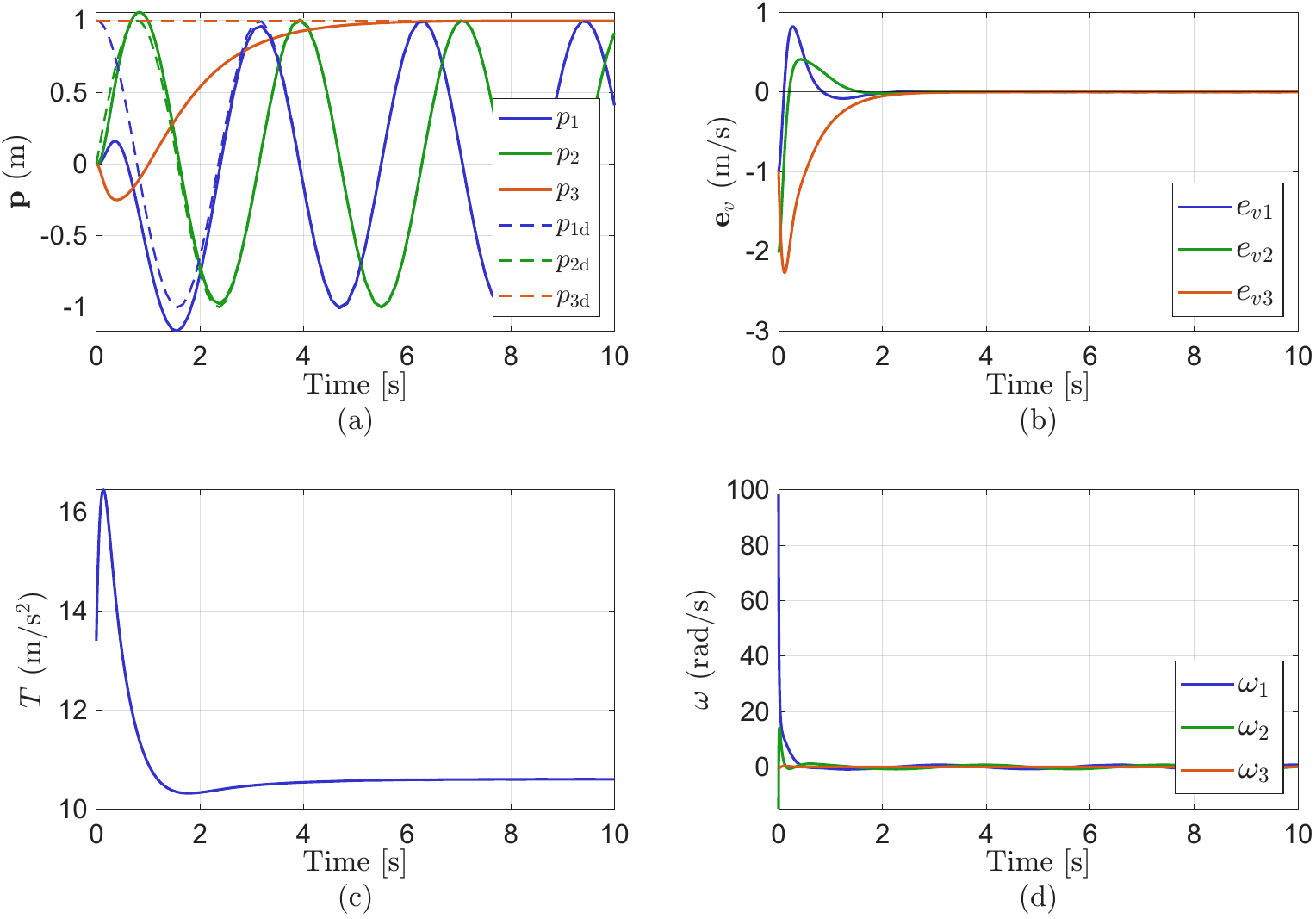}
    \vspace{-4mm}
    \caption{\textbf{Geometric} controller tracking. The plots have the same interpretation as those in Fig. \ref{fig:de_track}.
    % (a) shows quadrotor position $\bp$ (solid), the desired position $\bp_\mathrm{d}$ (dashed). (b) shows velocity errors, (c) shows the total thrust, and (d) shows input angular velocities $\bm \omega$.
    }
    \label{fig:geometric_track}
    \vspace{-2mm}
\end{figure}

\textbf{Safety Enforcement.}
Next, the objective is to enforce the position constraints $p_\mathrm{m} = \pm 0.8\,\mathrm{m}$, namely $\vert p_1 \vert \le p_\mathrm{m}$, and $\vert p_2 \vert \le p_\mathrm{m}$, $\forall t\ge t_0$, while tracking the desired reference trajectory $\bp_\mathrm{d}$. The position constraints are incorporated through the smooth combination \cite{TamasLCSS23}:
\begin{equation}
h(\bp) = -\tfrac{1}{\lambda_h}
\log\big(
e^{-\lambda_\mathrm{h}(p_\mathrm{m}^2 - p_1^2)}
+
e^{-\lambda_\mathrm{h}(p_\mathrm{m}^2 - p_2^2)}
\big),
\end{equation}
which defines the safe set $\mathcal{C}$ as in \eqref{eq:C-quadrotor}. A smooth safety filter \cite{CohenLCSS23,CohenARC24} is then applied to the nominal $\bv_{\rm{d}}$ to generate a safe and differentiable velocity-field that satisfies \eqref{eq:safe-desired-vector field}.

For these tests, the initial rotation error is set to $\theta_0=\tfrac{\pi}{2}$. The controller gains are selected as $k_1=k_2=20$, while the safety-filter parameters are set to $\alpha=7.5$, $\lambda_\mathrm{h} = 5$, $\epsilon=2$, and $\mu=10$. Fig. \ref{fig:de_safety} shows results for the the dynamic extension controller. For the geometric controller $k_1=k_2=20$, the safety-filter parameters are set to $\alpha=5$, $\lambda_\mathrm{h} = 2$, $\epsilon=2$, and $\mu=10$. Fig. \ref{fig:geometric_safety} shows the results for the geometric controller.
Simulations show that, by leveraging the proposed vector field tracking controllers, safety properties can be effectively encoded into the desired vector field using CBFs and subsequently transferred to the full quadrotor dynamics.

\begin{figure}[h]
    \centering
    % \vspace{-2mm}
    \includegraphics[width=0.8\columnwidth]{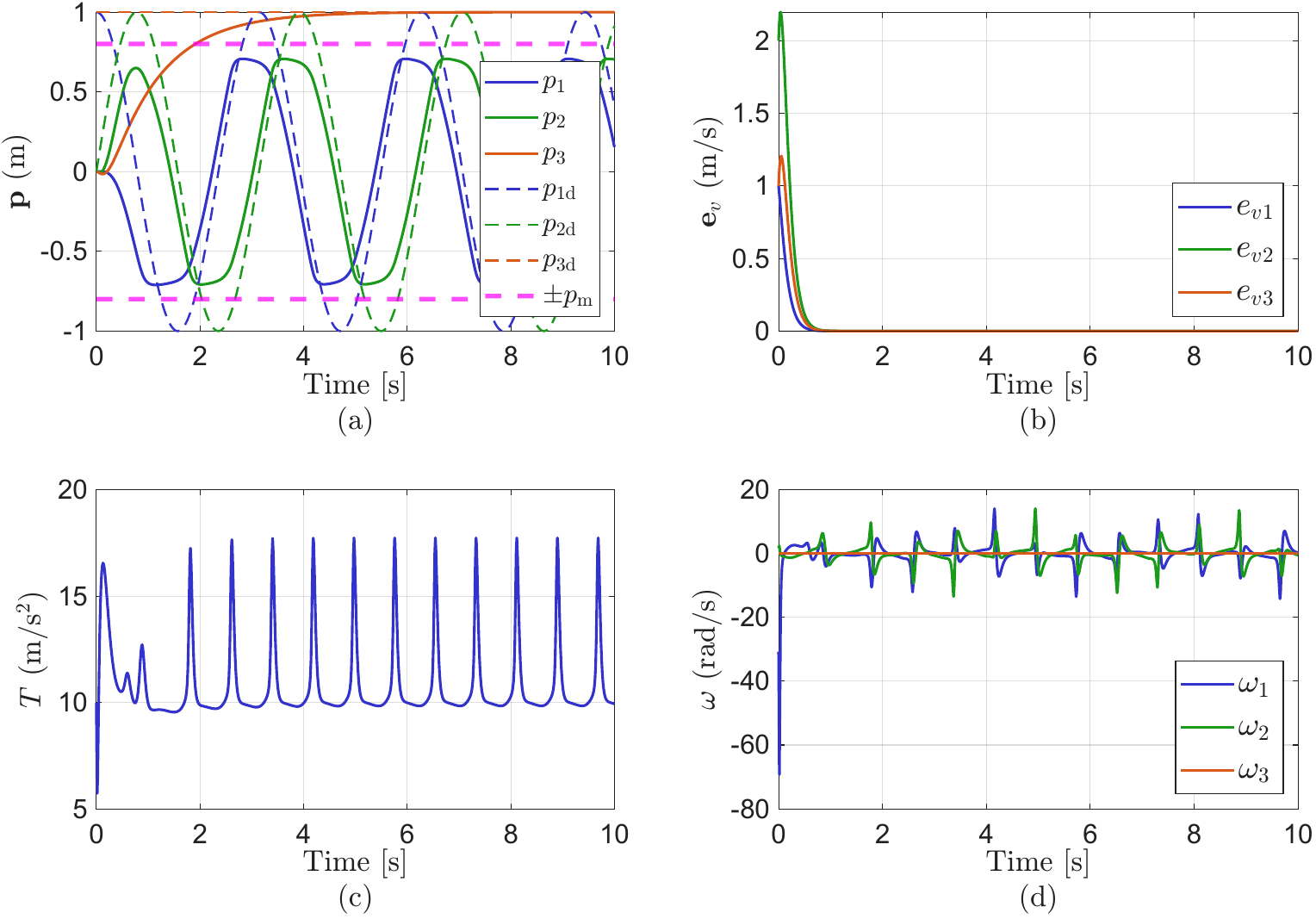}
    \vspace{-4mm}
    \caption{\textbf{DE} controller safe tracking. (a) shows quadrotor position $\bp$ (solid), the desired position $\bp_\mathrm{d}$ (dashed), and the constraint lines $\pm p_\mathrm{m}$ (hlines). (b) shows velocity errors, (c) shows the total thrust, and (d) shows input angular velocities $\bm \omega$.}
    \label{fig:de_safety}
    \vspace{-2mm}
\end{figure}

\begin{figure}[h]
    \centering
    % \vspace{-2mm}
    \includegraphics[width=0.8\columnwidth]{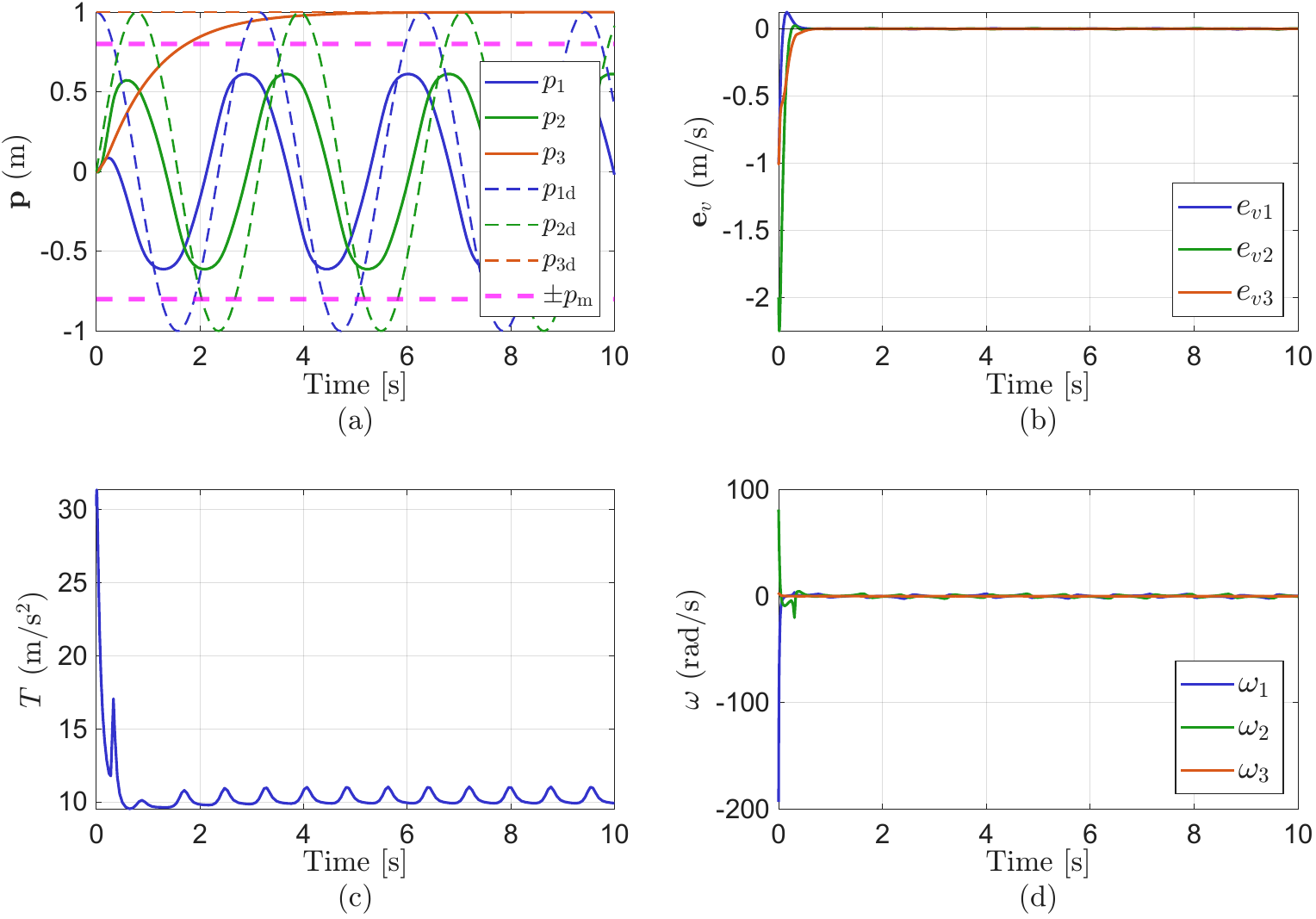}
    \vspace{-4mm}
    \caption{\textbf{Geometric} controller safe tracking. The interpretation of the plots are the same as those in Fig. \ref{fig:de_safety}.
    % (a) shows quadrotor position $\bp$ (solid), the desired position $\bp_\mathrm{d}$ (dashed), and the constraint lines $\pm p_\mathrm{m}$ (hlines). (b) shows velocity errors, (c) shows the total thrust, and (d) shows input angular velocities $\bm \omega$.
    }
    \label{fig:geometric_safety}
    \vspace{-2mm}
\end{figure}

\section{Experimental Validation}
In this section, we show the effectiveness of the proposed vector field tracking controllers and safety enforcement by conducting indoor flight experiments with a Crazyflie 2.1. A video of our experiments can be viewed at \cite{FlightTests}.

\textbf{Trajectory Tracking.}
First, we perform flight tests to show the vector field tracking performance. The drone is initialized at rest on level ground at the origin. The reference trajectory is $\bp_\mathrm{d}(t)=[\tfrac{1}{2}\cos{\beta t}\;\tfrac{1}{2} \sin{\beta t}\; 0.25]\T$, and $\beta = \tfrac{2\pi}{5}$. Then the desired vector field is $\bv_\mathrm{d}(\bp,t)=\dot \bp_\mathrm{d}-\lambda(\bp-\bp_\mathrm{d})$, and $\lambda=2.5$. For the DE controller $k_1=k_2=20$, and for the geometric controller $k_1=k_2=10$. 
Fig. \ref{fig:de_exp_track} shows the DE controller performance, and Fig. \ref{fig:geometric_exp_track} shows the geometric controller results, which indicate that 
the vector field tracking controllers track the desired trajectory.

\begin{figure}[h]
    \centering
    % \vspace{-2mm}
    \includegraphics[width=0.8\columnwidth]{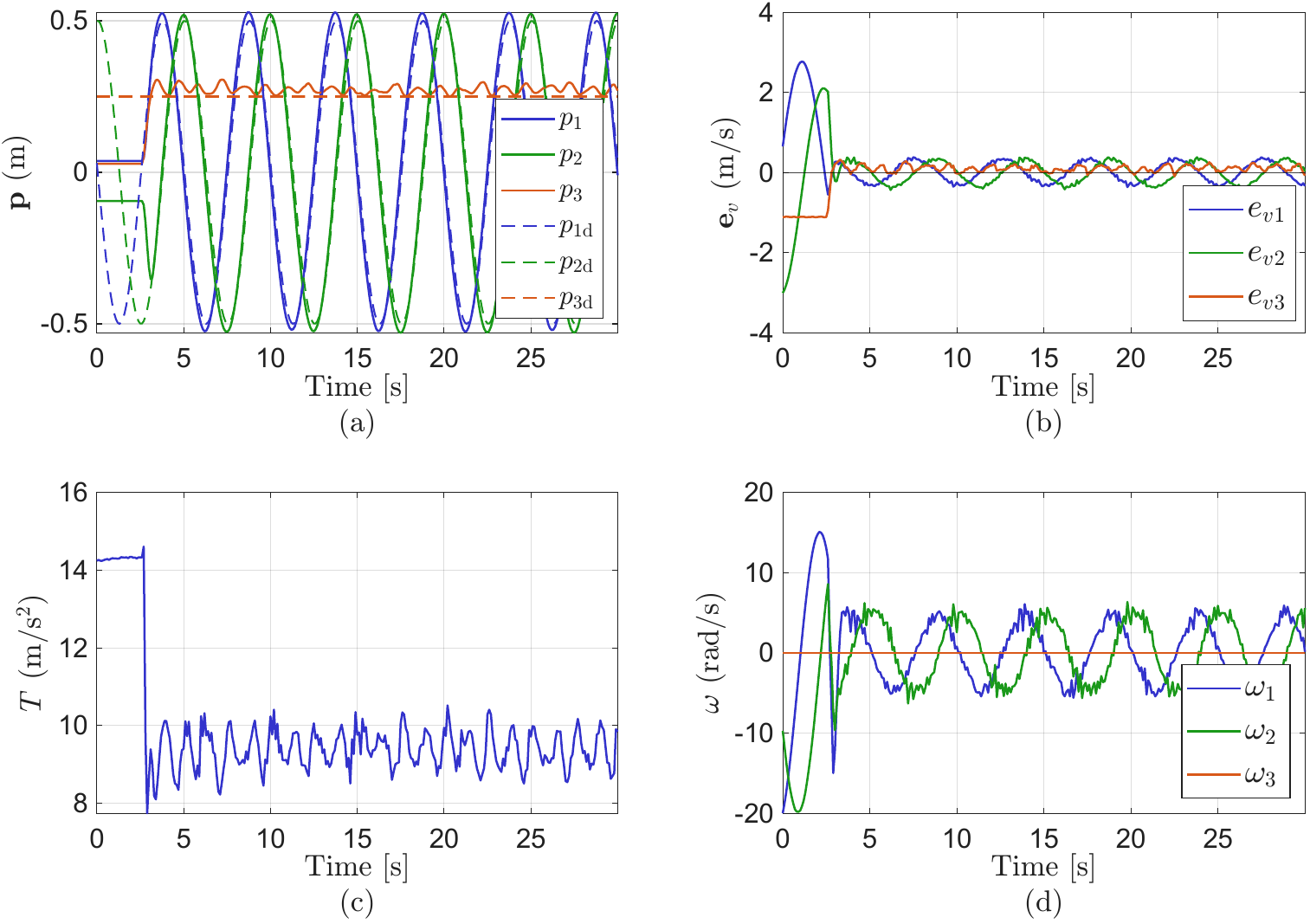}
    \vspace{-4mm}
    \caption{\textbf{DE} controller reference tracking experiment. (a) shows position $\bp$ (solid), the desired position $\bp_\mathrm{d}$ (dashed). (b) shows velocity errors, (c) shows the total thrust, and (d) shows input angular velocities $\bm \omega$.}
    \label{fig:de_exp_track}
    \vspace{-2mm}
\end{figure}

\begin{figure}[h]
    \centering
    % \vspace{-2mm}
    \includegraphics[width=0.8\columnwidth]{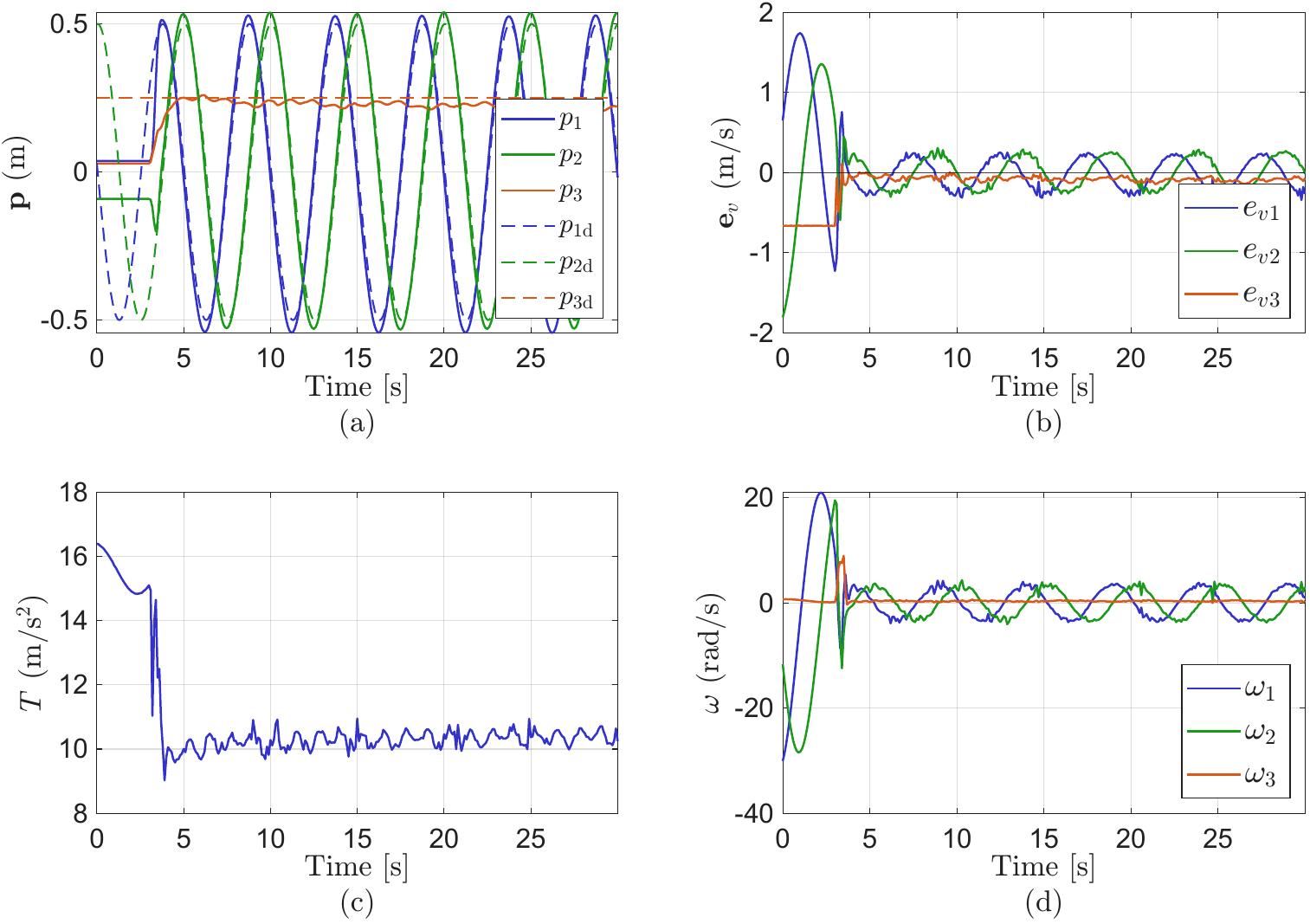}
    \vspace{-4mm}
    \caption{\textbf{Geometric} controller reference tracking experiment. 
    The interpretation of the plots are the same as in Fig. \ref{fig:de_exp_track}.
    % (a) shows position $\bp$ (solid), the desired position $\bp_\mathrm{d}$ (dashed). (b) shows velocity errors, (c) shows the total thrust, and (d) shows input angular velocities $\bm \omega$.
    }
    \label{fig:geometric_exp_track}
    \vspace{-6mm}
\end{figure}

\textbf{Safety Enforcement.}
We now show that imposing CBF conditions on the desired vector fields leads to the satisfaction of position constraints. 
% The reference trajectory is the same as above.
% $\bp_\mathrm{d}(t)=[\tfrac{1}{2}\cos{\beta t},\quad\tfrac{1}{2} \sin{\beta t},\quad 0.25]\T$, and $\beta = \tfrac{2\pi}{5}$. 
The objective is to enforce the position constraints $p_\mathrm{m} = \pm 0.4\,\mathrm{m}$, namely $\vert p_1 \vert \le p_\mathrm{m}$, and $\vert p_2 \vert \le p_\mathrm{m}$, while tracking the desired reference trajectory $\bp_\mathrm{d}$ as above. The CBF parameters are $\lambda_h=10$, $\alpha = 5$, $\epsilon=1$, and $\mu=10$. Fig. \ref{fig:de_exp_safety} shows enforcing position constraints with the DE controller, and Fig. \ref{fig:geometric_exp_safety} shows the geometric controller results. Results show that the desired safety behavior is successfully encoded into the desired vector fields and accurately tracked by the controllers.

\begin{figure}[h]
    \centering
    % \vspace{-2mm}
    \includegraphics[width=0.8\columnwidth]{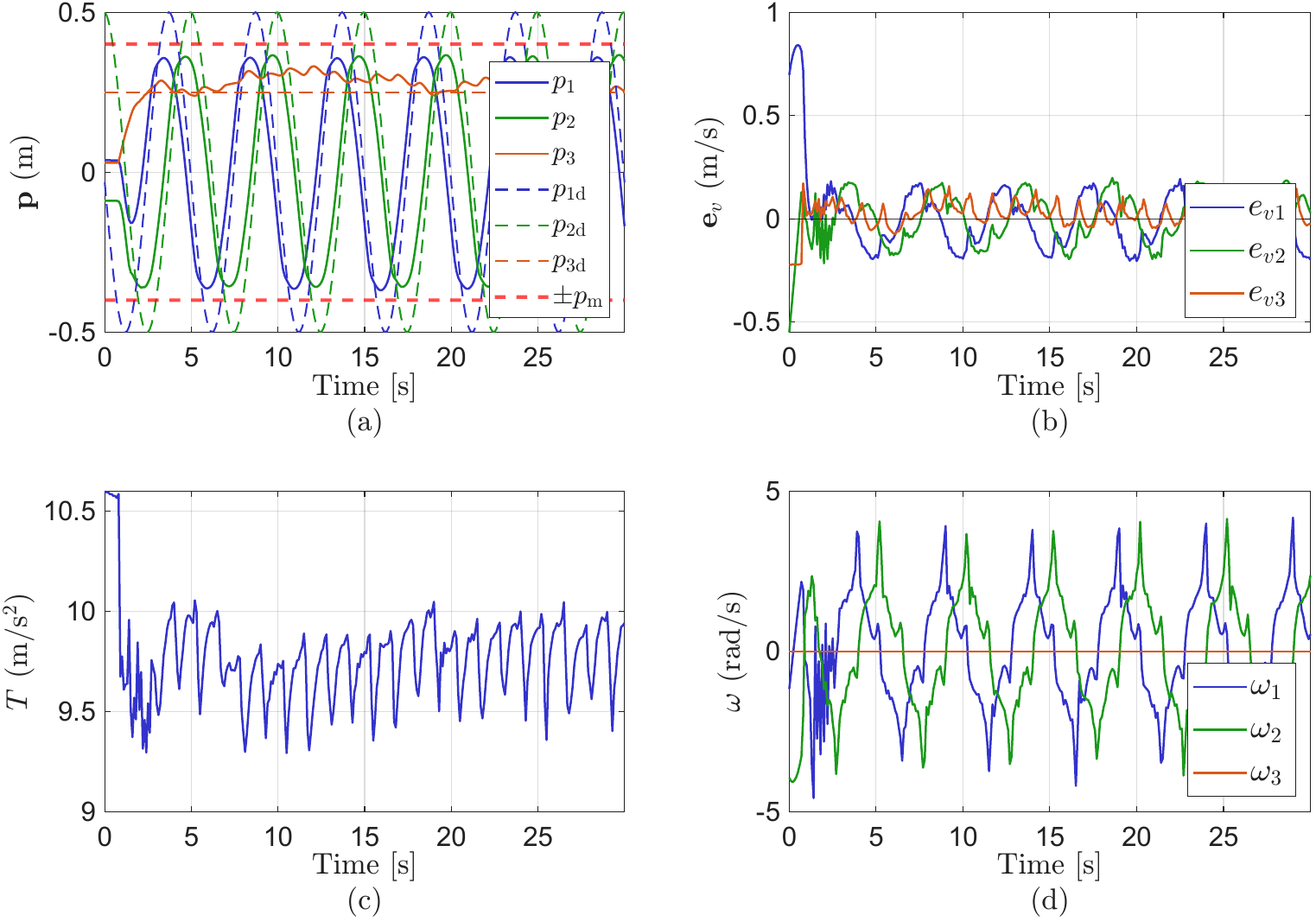}
    \vspace{-4mm}
    \caption{\textbf{DE} controller safe reference tracking experiment. (a) shows  position $\bp$ (solid), the desired position $\bp_\mathrm{d}$ (dashed), and the constraints $\pm p_\mathrm{m}$ (hlines). (b) shows velocity errors, (c) shows the total thrust, and (d) shows input angular velocities $\bm \omega$.}
    \label{fig:de_exp_safety}
    \vspace{-2mm}
\end{figure}

\begin{figure}[h]
    \centering
    % \vspace{-2mm}
    \includegraphics[width=0.8\columnwidth]{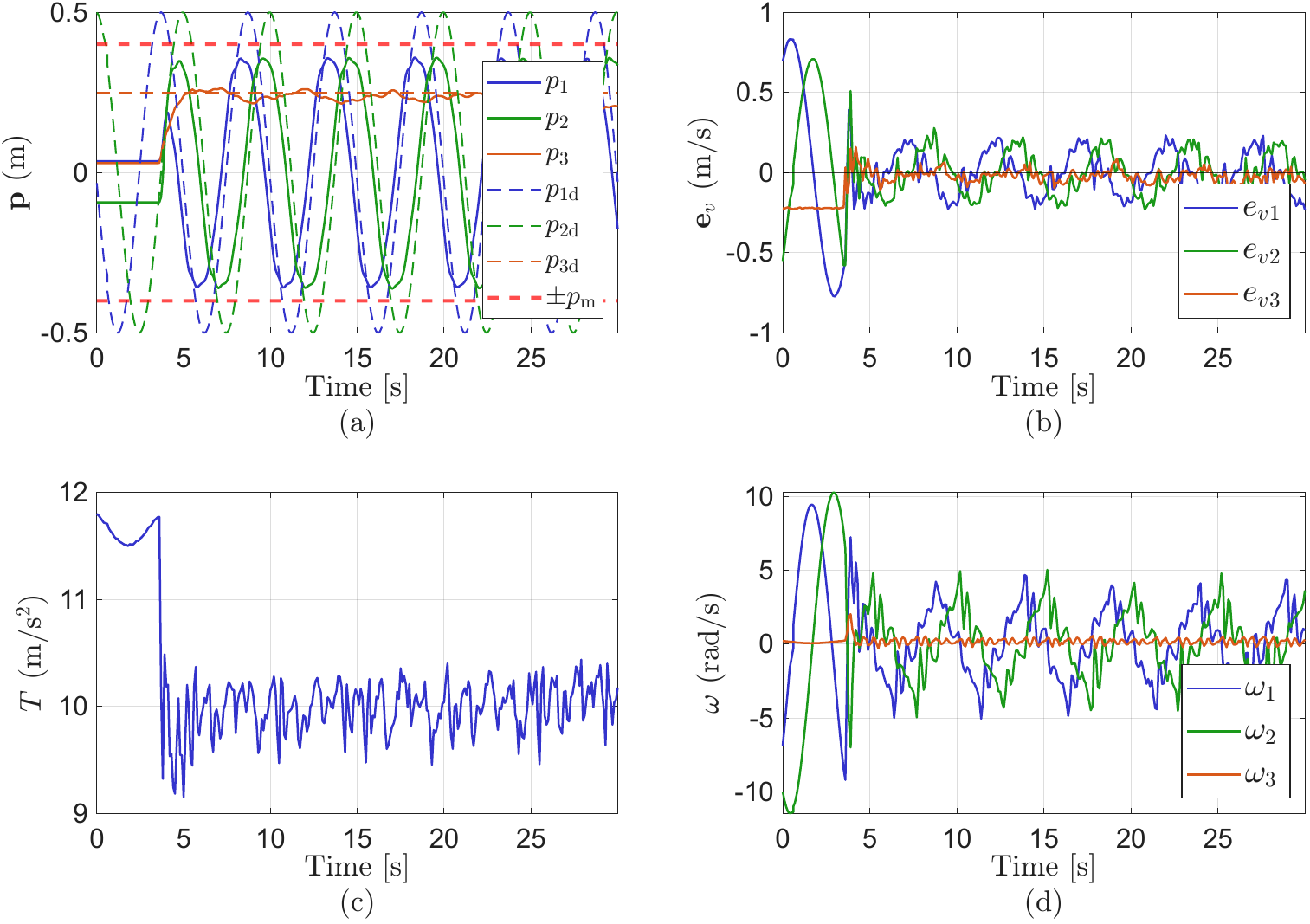}
    \vspace{-4mm}
    \caption{\textbf{Geometric} controller safe reference tracking experiment. 
    The interpretation of the plots are the same as those in Fig. \ref{fig:de_exp_safety}
    % (a) shows  position $\bp$ (solid), the desired position $\bp_\mathrm{d}$ (dashed), and the constraint lines $\pm p_\mathrm{m}$ (hlines). (b) shows velocity errors, (c) shows the total thrust, and (d) shows input angular velocities $\bm \omega$.
    }
    \label{fig:geometric_exp_safety}
    \vspace{-2mm}
\end{figure}

\section{Conclusion}
In this work, we proposed two vector field tracking controllers for quadrotors based on reduced-attitude representations. We established stability of the proposed dynamic extension and geometric controllers, followed by a framework for encoding safety into the desired vector fields. Numerical simulations and indoor flight experiments demonstrated the effectiveness of the proposed approaches.

\textbf{Acknowledgments.}
ChatGPT was used to assist in editing Fig. \ref{fig:exp_video} and the experiment video \cite{FlightTests}. We thank Massimiliano de Sa for the helpful discussions on CBFs on manifolds.

% \section*{Appendix}
% Another way to express $(\bR_3-\bR_\mathrm{3d})$ in terms of $\be_2= \bR_3 \times \bR_{3\mathrm{d}}$ to deal with the translational and rotational cross coupling term in \eqref{eq:dot-v1-ge} is to decompose it into tangential and radial  components as follows:
% \begin{equation}
%     \bR_3-\bR_\mathrm{3d}=[\bR_3]_\times \be_2+\frac{\be_2\T\be_2}{1+\bR_3\T\bR_{3\mathrm{d}}}\bR_3.
% \end{equation}
% \begin{proof}
% \begin{equation}
%     \bR_3\times\be_2 = \bR_3 \times (\bR_3\times\bR_{3\mathrm{d}}) = (\bR_3\T\bR_{3\mathrm{d}})\bR_3 - \bR_{3\mathrm{d}},
% \end{equation}
% so adding and subtracting $(\bR_3\T\bR_{3\mathrm{d}})\bR_3$ gives:
% \begin{equation}
%     \bR_3-\bR_{3\mathrm{d}} = \bR_3\times\be_2 + (1-\bR_3\T\bR_{3\mathrm{d}})\bR_3.
% \end{equation}
% Since $\be_2\T\be_2 = 1-(\bR_3\T\bR_{3\mathrm{d}})^2 = (1-\bR_3\T\bR_{3\mathrm{d}})(1+\bR_3\T\bR_{3\mathrm{d}})$, for $\bR_3\T\bR_{3\mathrm{d}}\neq-1$ we obtain:
% \begin{equation}
%     \bR_3-\bR_{3\mathrm{d}} = [\bR_3]_\times\be_2 + \frac{\be_2\T\be_2}{1+\bR_3\T\bR_{3\mathrm{d}}}\bR_3.
% \end{equation}
% \end{proof}

\bibliographystyle{ieeetr}
\bibliography{biblio}

\end{document}